\documentclass{conm-p-l}

\newtheorem{theorem}{Theorem}[section]
\newtheorem{lemma}[theorem]{Lemma}

\theoremstyle{definition}

\newtheorem{proposition}[theorem]{Proposition}

\theoremstyle{remark}
\newtheorem{remark}[theorem]{Remark}

\numberwithin{equation}{section}

\begin{document}
\title{On a Random Matrix Models of Quantum Relaxation}
\author{J. L. Lebowitz}
\address{Department of Mathematics, Rutgers University, USA}
\thanks{Department of Physics of Rutgers University}

\author{A. Lytova}


\author{L. Pastur}
\address{Mathematical Division, Institute for Low Temperatures \\
Kharkiv, Ukraine}
\subjclass{Primary , ; Secondary , }


\keywords{Random matrix}

\begin{abstract}
In paper \cite{Le-Pa:04} two of us (J.L. and L.P.) considered a
matrix model for a two-level system interacting with a $n\times n$
reservoir and assuming that the interaction is modelled by a
random matrix. We presented there a formula for the reduced
density matrix in the limit $n\rightarrow \infty $ as well as
several its properties and asymptotic forms in various regimes. In
this paper we give the proofs of assertions, announced in
\cite{Le-Pa:04}. We present also a new fact about the model (see
Theorem \ref{t:sa}) as well as additional discussions of topics of
\cite{Le-Pa:04}
\end{abstract}

\maketitle

\section{Introduction}

\label{s:intr}

The model considered in \cite{Le-Pa:04} can be viewed as a random
matrix version of the spin-boson model, widely used in studies of
open quantum systems (see e.g. review works \cite{Li-Co,We:99}
and references therein). We mention here that one of the first
models of this type, namely the model where the classical system
is represented by a harmonic oscillator coupled linearly with the
oscillator reservoir, was considered by N. Bogolyubov in 1945
\cite{Bo:45}, Chapter IV.

We recall now the model, proposed and discussed in \cite{Le-Pa:04}. Let $%
h_{n}$ be a Hermitian $n\times n$ matrix with eigenvalues $%
E_{j}^{(n)},\;j=1,...,n$. We characterize the spectrum of $h_{n}$
by its normalized counting measure of eigenvalues
\begin{equation}
\nu _{0}^{(n)}(\Delta )=n^{-1}\sum_{j=1}^{n}\chi _{\Delta
}(E_{j}^{(n)});\quad \int_{-\infty }^{\infty }\nu
_{0}^{(n)}(dE)=1,  \label{nun}
\end{equation}%
where $\chi _{\Delta }$ is the indicator of an interval $\Delta
\subset \mathbb{R}$. We assume that $\nu _{0}^{(n)}$ converges
weakly as $n\rightarrow \infty $ to a limiting probability
measure $\nu _{0}$, i.e. that for any bounded and continuous
function $\varphi :\mathbb{R}\rightarrow \mathbb{R}$ we have:
\begin{equation}
\lim_{n\rightarrow \infty }\int_{-\infty }^{\infty }\nu
_{0}^{(n)}(dE)\varphi (E)=\int_{-\infty }^{\infty }\nu
_{0}(dE)\varphi (E),\quad \;\int_{-\infty }^{\infty }\nu
_{0}(dE)=1.  \label{nu}
\end{equation}%
Let $w_{n}$ be a Hermitian $n\times n$ random matrix, whose
probability density is
\begin{equation}
Q_{n}^{-1}\exp \left\{ -\mathrm{Tr\;}w^{2}_n/2\right\} ,
\label{GUE}
\end{equation}%
where $Q_{n}$ is the normalization constant. In other words, the entries $%
w_{jk},\;1\leq j\leq k\leq n$ of the matrix $w_{n}$ are
independent Gaussian random variables with
\begin{equation}
\left\langle w_{jk}\right\rangle =0,\;\left\langle
w_{jj}{}^{2}\right\rangle =1,\;j,k=1,...,n,\;\left\langle (\Re
w_{jk})^{2}\right\rangle =\left\langle (\Im
w_{jk})^{2}\right\rangle =\frac{1}{2},\;j\neq k, \label{G}
\end{equation}%
where the symbol $\left\langle ...\right\rangle $ denotes here
and below the expectation with respect to the distribution
(\ref{GUE}). This probability distribution is known as the
Gaussian Unitary Ensemble (GUE) \cite{Me:91}.

We define the Hamiltonian of our composite system $\mathcal{S}%
_{2,n} $ as a random $2n\times 2n$ matrix of the form
\begin{equation}
H^{(n)}=s\sigma ^{z}\otimes \mathbf{1}_{n}+\mathbf{1}_{2}\otimes
h_{n}+v\sigma ^{x}\otimes w_{n}/n^{1/2},  \label{Ham}
\end{equation}%
where $\mathbf{1}_{l}$ ($l=2,n$) is the $l\times l$ unit matrix,
$\sigma ^{z} $ and $\sigma ^{x}$ are the Pauli matrices
\begin{equation*}
\sigma ^{x}=\left(
\begin{array}{ll}
0 & 1 \\
1 & 0%
\end{array}%
\right) ,\;\sigma ^{z}=\left(
\begin{array}{ll}
1 & 0 \\
0 & -1%
\end{array}%
\right) ,
\end{equation*}%
the symbol $\otimes $ denotes the tensor product, and $s$ and $v$
are positive parameters.
The first term in (\ref{Ham}) is the Hamiltonian of the two-level sytem $%
\mathcal{S}_{2}$, the second term is the Hamiltonian of the
$n$-level system (reservoir) $\mathcal{S}_{n}$, and the third
term is an interaction between them. Thus $s$ determines the
energy scale of the isolated small system ($2s$ is its level
spacing), and $v$ plays the role of the coupling constant between
$\mathcal{S}_{2}$ and $\mathcal{S}_{n}$. We write the Hamiltonian
$H^{(n)}$ in the form
\begin{equation}
H^{(n)}=H_{0}^{(n)}+M^{(n)},  \label{Hn}
\end{equation}%
where%
\begin{equation*}
H_{0}^{(n)}=s\sigma ^{z}\otimes
\mathbf{1}_{n}+\mathbf{1}_{2}\otimes h_{n},\quad M^{(n)}=v\sigma
^{x}\otimes w_{n}/n^{1/2},
\end{equation*}%
and choose the basis in $\mathbb{C}^{2}\mathbb{\otimes C}^{n}$ in
which the matrix $H_{0}^{(n)}$ is diagonal:
\begin{equation}
(H_{0}^{(n)})_{\alpha j,\beta k}=\lambda _{\alpha j}^{(n)}\delta
_{\alpha \beta }\delta _{jk},\quad \lambda _{\alpha
j}^{(n)}=E_{j}^{(n)}+\alpha s,\quad \alpha ,\beta =\pm ,\quad
j,k=1,..,n.  \label{Hnmat}
\end{equation}%
Assume that at $t=0$ the density matrix of the composite system $\mathcal{S}%
_{2,n}$ is
\begin{equation}
\mu _{m}^{(n)}(E_{k}^{(n)},0)=\rho (0)\otimes P_{k},  \label{ini}
\end{equation}%
where $\rho (0)$ is a $2\times 2$ positive definite matrix of
unit trace and $P_{k}$ is the projection on the state of energy
$E_{k}^{(n)}$ of the reservoir. Let $\mu ^{(n)}(t)$ be the
density matrix of the composite system $\mathcal{S}_{2,n}$ at
time $t$, corresponding to the initial density matrix $\mu
_{m}^{(n)}(0)$ of (\ref{ini}):
\begin{equation}
\mu ^{(n)}(t)=U^{(n)}(-t)\mu
_{m}^{(n)}(E_{k}^{(n)},0)U^{(n)}(t),\quad
U^{(n)}(t)=e^{itH^{(n)}}.  \label{mun}
\end{equation}%
Then the reduced density matrix of the small system is defined as
\begin{equation}
\widehat{\rho }_{\alpha ,\delta
}^{(n)}(E_{k}^{(n)},t)=\sum_{j=1}^{n}\mu _{\alpha j,\delta
j}^{(n)}(t),\quad \alpha ,\delta =\pm ,  \label{rdmo}
\end{equation}%
i.e. $\widehat{\rho }^{(n)}$ is obtained from the density matrix
(\ref{mun}) of the whole composite (closed) system by tracing out
the reservoir degrees of freedom. It follows from Theorem
\ref{t:sa} below that the variance of
the reduced density matrix vanishes as $n\rightarrow \infty $, i.e. that $%
\widehat{\rho }^{(n)}$ is selfaveraging. This allows us to
confine ourselves to the study of the mean reduced density matrix
$\rho ^{(n)}(E_{k}^{(n)},t)$:
\begin{equation}
{\rho }_{\alpha ,\delta
}^{(n)}(E_{k}^{(n)},t)=\sum_{j=1}^{n}\left\langle \mu _{\alpha
j,\delta j}^{(n)}(t)\right\rangle =\sum_{\beta ,\gamma =\pm
}T_{\alpha \beta \gamma \delta }^{(n)}(E_{k}^{(n)},t)\rho _{\beta
,\gamma }(0),  \label{ron}
\end{equation}%
where
\begin{equation}
T_{\alpha \beta \gamma \delta
}^{(n)}(E_{k}^{(n)},t)=\sum_{j=1}^{n}\left\langle U_{\alpha
j,\beta k}^{(n)}(-t)+U_{\gamma k,\delta j}^{(n)}(t)\right\rangle
\label{per}
\end{equation}%
is the \textquotedblright transfer\textquotedblright\ matrix, an
analog of the influence functional by Feynman-Vernon \cite{We:99}.

Notice that we can equally consider the factorized initial
condition $\mu _{c}^{(n)}(\beta ,0)$ in which the microcanonical
distribution $P_{k}$ of the reservoir is replaced by its
canonical distribution $e^{-\beta
H_{n}^{(0)}}/Z_{0}^{(n)}$. We have evidently%
\begin{equation}
\mu _{c}^{(n)}(\beta ,0)=\sum_{k=1}^{n}e^{-\beta E_{k}^{(n)}}\mu
_{m}^{(n)}(E_{k}^{(n)},0)/Z_{0}^{(n)}.  \label{inic}
\end{equation}%
It is also easy to write the corresponding reduced density matrix.

\section{Selfaveraging of Reduced Density Matrix}

\label{s:sa}

\begin{theorem}
\label{t:sa} Let $\widehat{\rho }^{(n)}(E_{k}^{(n)},t)$ be
the reduced density matrix (\ref{rdmo}) of the composite system $%
S_{2,n}=S_{2}+S_{n}$, given by (\ref{nun}) -- (\ref{mun}). Then
we have
\begin{equation}
\mathbf{Var}\biggl \{\widehat{\rho }_{\alpha ,\delta }^{(n)}(E_{k}^{(n)},t)\biggr \}%
\leq \frac{8v^{2}t^{2}}{n},\;\alpha ,\delta =\pm .  \label{var}
\end{equation}
\end{theorem}

\noindent To prove the theorem we use the following facts.

\begin{proposition}
\label{p:PN} (Poincare-Nash inequality). If a random Gaussian vector $%
X=\{\xi _{j}\}_{j=1}^{p}$ satisfies conditions $\langle \xi
_{j}\rangle =0$,
$\langle \xi _{j}\bar{\xi _{k}}\rangle =C_{jk}$, $j,k=1,..,p$, and functions $%
\Phi _{1,2}:\;\mathbb{R}^{p}\rightarrow \mathbb{C}$ have bounded
partial derivatives, then
\begin{align}
\mathbf{Cov}\{\Phi _{1},\Phi _{2}\}&:=\langle(\Phi _{1},\Phi
_{2})\rangle-\langle\Phi _{1}\rangle\langle\Phi _{2}\rangle
\label{PN}
\\
&\leq \langle (C\nabla \Phi _{1},\nabla \Phi _{1})\rangle
^\frac{1}{2}\langle (C\nabla \Phi _{2},\nabla \Phi _{2})\rangle
^\frac{1}{2}, \notag
\end{align}
where
\begin{equation*}
(C\nabla \Phi ,\nabla \Phi )=\sum_{j,k=1}^{p}C_{jk}(\nabla \Phi
)_{j}\bar{(\nabla \Phi )_{k}}.
\end{equation*}
\end{proposition}

\noindent For the proof of the inequality see e.g. \cite{Bo:99},
Theorem 1.6.4.

\begin{proposition}
\label{p:Duh} (Duhamel formula). If $M_{1}$, $M_{2}$ are $n\times
n$-matrices, then
\begin{equation}
e^{(M_{1}+M_{2})t}=e^{M_{1}t}+%
\int_{0}^{t}e^{M_{1}(t-s)}M_{2}e^{(M_{1}+M_{2})s}ds.  \label{D}
\end{equation}
\end{proposition}

\noindent The proof is elementary. Notice, that Duhamel formula
allows us to obtain
the derivative of the matrix $U^{(n)}(t)$ with respect to the entry $%
w_{lm},\;l,m=1,...,n$ of the matrix $w_n$ in (\ref{Ham}):
\begin{equation}
\frac{\partial U_{\alpha j,\beta k}^{(n)}(t)}{\partial w_{lm}}=\frac{iv}{%
\sqrt{n}}\int_{0}^{t}\sum_{\kappa =\pm }U_{\alpha j,\kappa
l}^{(n)}(t-s)U_{-\kappa m,\beta k}^{(n)}(s)ds.  \label{DF}
\end{equation}

\begin{proof} (of the Theorem \ref{t:sa}). By using the
Poincare-Nash inequality (\ref{PN}) with%
\begin{equation*}
\Phi _{1}(X)=\Phi _{2}(X)=\sum_{j=1}^{n}U_{\gamma k,\delta
j}^{(n)}(t)U_{\alpha j,\beta
k}^{(n)}(-t),\;X=\{w_{lm}\}_{l,m=1}^{n},\;C=\mathbf{1}_{n^{2}},
\end{equation*}
differentiation formula (\ref{DF}), and then Schwartz inequality,
we obtain

\begin{align}
\mathbf{Var}\biggl \{\sum_{j}&U_{\gamma k,\delta j}^{(n)}(t)
U_{\alpha
j,\beta k}^{(n)}(-t)\biggr \}  \label{v} \\
& \leq \frac{v^{2}}{n}\biggl \langle\sum_{l,m}\biggl |\sum_{j}\int_{0}^{t}%
\sum_{\kappa }U_{\gamma k,\kappa l}^{(n)}(t-s)U_{-\kappa m,\delta
j}^{(n)}(s)ds\;U_{\alpha j,\beta k}^{(n)}(-t)  \notag \\
& +\sum_{j}U_{\gamma k,\delta j}^{(n)}(t)\int_{0}^{t}\sum_{\kappa
}U_{\alpha
j,\kappa l}^{(n)}(-t+s)U_{-\kappa m,\beta k}^{(n)}(-s)ds\biggr |^{2}%
\biggl\rangle  \notag \\
& \leq \frac{2v^{2}t}{n}\biggl \langle\int_{0}^{t}\sum_{\kappa ,l}\bigl |%
U_{\gamma k,\kappa l}^{(n)}(t-s)\bigr |^{2}\sum_{\kappa ,m}\biggl |%
\sum_{j}U_{-\kappa m,\delta j}^{(n)}(s)U_{\alpha j,\beta k}^{(n)}(-t)\biggr |%
^{2}ds  \notag \\
& +\int_{0}^{t}\sum_{\kappa ,l}\biggl |\sum_{j}U_{\alpha j,\kappa
l}^{(n)}(-t+s)U_{\gamma k,\delta j}^{(n)}(t)\biggr |^{2}\sum_{\kappa ,m}%
\bigl |U_{-\kappa m,\beta k}^{(n)}(-s)\bigr |^{2}ds\biggl\rangle,
\notag
\end{align}%
and here and below all the sums over the Latin indices will be from $1$ to $%
n $, and the sum over the Greek indices will be over $\pm $.
Notice that
\begin{equation}
\sum_{\kappa ,m}\bigl |U_{-\kappa m,\beta k}^{(n)}(-s)\bigr |%
^{2}=\sum_{\kappa ,l}\bigl |U_{\gamma k,\kappa l}^{(n)}(t-s)\bigr
|^{2}=1, \label{U}
\end{equation}%
and
\begin{align}
\sum_{\kappa ,l}\biggl |\sum_{j}U_{\alpha j,\kappa
l}^{(n)}(-t+s)U_{\gamma k,\delta j}^{(n)}&(t)\biggr
|^{2}=\sum_{j_{1},j_{2}}\biggl (\sum_{\kappa ,l}U_{\alpha
j_{1},\kappa l}^{(n)}(-t+s)U_{\kappa l,\alpha
j_{2}}^{(n)}(t-s)\biggr )  \label{UU}
\\
&\times U_{\gamma k,\delta
j_{1}}^{(n)}(t)U_{\delta j_{2},\gamma k}^{(n)}(-t)=\sum_{j_{1}}\bigl |%
U_{\gamma k,\delta j_{1}}^{(n)}(t)\bigr |^{2}\leq 1.  \notag
\end{align}%
Hence, we have by (\ref{v}), (\ref{U}), and (\ref{UU}):%
\begin{equation*}
\mathbf{Var}\biggl \{\sum_{j}U_{\gamma k,\delta
j}^{(n)}(t)U_{\alpha j,\beta k}^{(n)}(-t)\biggr \}\leq
\frac{4v^{2}t^{2}}{n},\quad \alpha ,\delta =\pm .
\end{equation*}%
Now, taking into account (\ref{rdmo}) and the fact that $\rho (0)$ is a $%
2\times 2$ positive definite matrix and of unit trace, we obtain (\ref{var}%
).
\end{proof}

\section{Equilibrium Properties}

\label{s:equi}

We begin by considering the equilibrium (time independent)
microcanonical density matrix of the composite system
${}\mathcal{S}_{2,n}$:
\begin{equation}
\Omega (\lambda )=\delta (\lambda -H^{(n)})\left/
\mathrm{Tr}\;\delta (\lambda -H^{(n)})\right. .  \label{med}
\end{equation}%
Following a standard prescription of statistical mechanics, we
will replace the Dirac delta-function in (\ref{med}) by the
function $(2\varepsilon )^{-1}\chi _{\varepsilon }$, where $\chi
_{\varepsilon }$ is the indicator of the interval $(-\varepsilon
,\varepsilon )$, and $\varepsilon \ll \lambda $. Then the reduced
microcanonical density matrix, i.e. the microcanonical density
matrix of $\mathcal{S}_{2,n}$, traced with respect to the states
of ${}\mathcal{S}_{n}$, is the $2\times 2$ matrix of the form
\begin{equation}
\omega ^{(n)}(\lambda )=\frac{\overline{\nu }^{(n)}(\lambda
)}{\sum_{\delta =\pm }\overline{\nu }_{\delta ,\delta
}^{(n)}(\lambda )},  \label{mer}
\end{equation}%
where
\begin{equation}
\overline{\nu }_{\alpha \gamma }^{(n)}(\lambda )=(2\varepsilon
n)^{-1}\sum_{j=1}^{n}\chi _{\varepsilon }(\lambda
-H^{(n)})_{\alpha j,\gamma j}.  \label{mer1}
\end{equation}%
The corresponding canonical distribution of the composite system
is
\begin{equation}
e^{-\beta H_{n}}\left/ \mathrm{Tr}\;e^{-\beta H_{n}}\right. ,
\label{ced}
\end{equation}%
and the reduced distribution of the small system is
\begin{equation}
\frac{\int_{-\infty }^{\infty }e^{-\beta \lambda }\nu ^{(n)}(d\lambda )}{%
\sum_{\delta =\pm }\int_{-\infty }^{\infty }e^{-\beta \lambda
}\nu _{\delta ,\delta }^{(n)}(d\lambda )},  \label{cer}
\end{equation}%
where (cf (\ref{mer1}))
\begin{equation}
\nu ^{(n)}(\Delta )=\{\nu _{\alpha \gamma }^{(n)}(\Delta
)\}_{\alpha ,\gamma =\pm },\quad \nu _{\alpha \gamma
}^{(n)}(\Delta )=n^{-1}\sum_{j=1}^{n}\chi _{\Delta
}(H^{(n)})_{\alpha j,\gamma j},  \label{nuen}
\end{equation}%
and $\chi _{\Delta }$ is the indicator of an interval $\Delta $
of the spectral axis.

\begin{theorem}
\label{t:mom1} Consider the $2\times 2$ matrix measure $\nu ^{(n)}$ of (\ref%
{nuen}). Then

(i) there exists non-random diagonal $2\times 2$ matrix measure%
\begin{equation*}
\nu =\{\nu _{\alpha }\delta _{\alpha \gamma }\}_{\alpha ,\gamma
=\pm }
\end{equation*}%
such that the weak convergence:%
\begin{equation}
\lim_{n\rightarrow \infty }\nu ^{(n)}=\nu  \label{nunu}
\end{equation}%
holds with probability 1;

(ii) if
\begin{equation*}
f_{\alpha }(z)=\int_{-\infty }^{\infty }\frac{\nu _{\alpha }(d\lambda )}{%
\lambda -z},\quad\Im z\neq 0,\label{ST}
\end{equation*}%
is the Stieltjes transform of $\nu _{\alpha }$, and $\nu _{0}$ is
defined by (\ref{nu}), then\ the pair $f_{\alpha }(z),\;\alpha
=\pm $ is a unique solution of the system of two coupled
functional equations
\begin{equation}
f_{\alpha }(z)=\int_{-\infty }^{\infty }\frac{\nu
_{0}(dE)}{E+s\alpha -z-v^{2}f_{-\alpha }(z)},\;\alpha =\pm
\label{falf}
\end{equation}%
in the class of functions analytic for $\Im z\neq 0$, and
satisfying the condition $\Im f_{\alpha }(z)\cdot \Im z>0,\;\Im
z\neq 0$;

(iii) nonnegative measures $\nu _{\alpha },\;\alpha =\pm $ \ $\
$have the unit total mass, $\nu _{\alpha }(\mathbb{R})=1$, and if
the measure $\nu_0$ of (\ref{nu}) is absolute continuous and
$\underset{\lambda\in\mathbb{R}}{\sup}\;\nu_0^{\prime
}(\lambda)<\infty$, then $\nu _{\alpha }$, $\alpha =\pm$ are also
absolute continuous, and we have
\begin{equation}
\nu _{\alpha }^{\prime }(\lambda )\leq
\underset{\mu\in\mathbb{R}}{\sup}\;\nu_0^{\prime }(\mu);
\label{nupr}
\end{equation}

(iv) for any $\lambda \in \mathbb{R}$ with probability 1 there
exists the limit of
the reduced microcanonical distribution%
\begin{equation}
\lim_{n\rightarrow \infty }\omega ^{(n)}=\omega ,\label{merl}
\end{equation}%
where%
\begin{equation}
\omega (\lambda )=\frac{\overline{\nu }(\lambda )}{\sum_{\delta =\pm }%
\overline{\nu }_{\delta ,\delta }(\lambda )},\label{merl1}
\end{equation}%
and%
\begin{equation}
\overline{\nu }_{\alpha \gamma }=\delta _{\alpha ,\gamma }\overline{\nu }%
_{\alpha },\quad \overline{\nu }_{\alpha }(\lambda )=(2\varepsilon
)^{-1}\int_{\lambda -\varepsilon }^{\lambda +\varepsilon }\nu
_{\alpha }(d\mu ),  \label{nuab}
\end{equation}%
analogous formulas are also valid for the limits of the reduced
canonical distribution (\ref{cer}).
\end{theorem}

\begin{remark}
\label{r:SF} The limiting measures $\nu _{\alpha },\;\alpha =\pm
$ can be found from their Stieltjes transforms $f_{\alpha
},\;\alpha =\pm $ via the
inversion formula \cite{Ak-Gl:93}:%
\begin{equation}
\nu _{\alpha }(\triangle )=\pi ^{-1}\lim_{\tau \rightarrow
0}\int_{\triangle }\Im f_{\alpha }(\lambda +i\tau )d\lambda .
\label{nal}
\end{equation}
\end{remark}

\noindent To prove the theorem we need the following auxiliary
fact.
\begin{proposition}
\label{p:diff} Let $\Phi $ be a $\mathbf{C}^{1}$ function of
$n\times n$ hermitian matrix, bounded together with its
derivatives. Then we have for the GUE matrix $w_n$ of (\ref{GUE}):
\begin{equation}
\bigl\langle\frac{\partial \Phi (w_n)}{\partial w_{jk}}\bigr\rangle=%
\bigl\langle\Phi (w_n)w_{kj}\bigr\rangle.  \label{wkj}
\end{equation}
\end{proposition}

\noindent The proof of proposition follows from
(\ref{GUE})--(\ref{G}) and the integration by parts formula.

\medskip

\begin{proof} (of the Theorem \ref{t:mom1}). Denote
\begin{equation}
G^{(n)}(z)=(H^{(n)}-z)^{-1},\;\Im z\neq 0  \label{resHn}
\end{equation}%
the resolvent of (\ref{Ham}) and set
\begin{equation}
g_{\alpha \gamma }^{(n)}(z)=n^{-1}\sum_{j=1}^{n}G_{\alpha j,\gamma
j}^{(n)}(z).  \label{gn}
\end{equation}%
It follows from the spectral theorem for Hermitian matrices that
$g_{\alpha \gamma }^{(n)}$ is the Stieltjes transform of $\nu
_{\alpha \gamma }^{(n)}$ and in view of the one-to-one
correspondence between measures and their Stieltjes transforms
(see \cite{Ak-Gl:93}, Section 59) to prove the weak convergence
(\ref{nunu}) with probability 1 it suffices to prove that with
probability 1 $g_{\alpha \gamma }^{(n)}$ converges to $\delta
_{\alpha \gamma }f_{\alpha }$ uniformly on a compact set of
$\mathbb{C}\backslash
\mathbb{R}$. Denote%
\begin{equation}
f_{\alpha \gamma }^{(n)}(z):=\bigl\langle g_{\alpha \gamma }^{(n)}(z)%
\bigr\rangle=n^{-1}\sum_{j=1}^{n}\bigl\langle G_{\alpha j,\gamma j}^{(n)}(z)%
\bigr\rangle.  \label{fag}
\end{equation}%
%
%
For further purposes it is convenient to start by considering the
functions
\begin{equation}
u_{\alpha \gamma }^{(n)}(t)=n^{-1}\sum_{j=1}^{n}\bigl\langle
U_{\alpha j,\gamma j}^{(n)}(t)\bigr\rangle,  \label{unUn}
\end{equation}%
where the matrix $U^{(n)}(t)$ is defined in (\ref{mun}). By the
spectral theorem for Hermitian matrices $u_{\alpha \gamma }^{(n)}
$ is the Fourier transfrom of $\nu _{\alpha \gamma }^{(n)}$ and
$f_{\alpha \gamma }^{(n)}(z)$ is the generalized Fourier
transform (see e.g. \cite{Ti:86}) of $u_{\alpha \gamma
}^{(n)}(t)$:
\begin{equation}
f_{\alpha \gamma }^{(n)}(z)=i^{-1}\int_{0}^{\infty
}e^{-izt}u_{\alpha \gamma }^{(n)}(t)dt,\quad \Im z<0.  \label{Fur}
\end{equation}%
Notice that the matrix $\bigl\langle U^{(n)}(t)\bigr\rangle$ is
diagonal with respect to the Latin indices. Indeed, since $w_n$
in (\ref{Ham}) is the GUE random matrix whose probability \ law
(\ref{GUE}) is unitary invariant, we have for any unitary
$n\times n$-matrix $U$:
\begin{equation}
\bigl\langle\exp \{itH^{(n)}\}\bigr\rangle=\bigl\langle\exp \{it\bigl (%
H_{0}^{(n)}+vn^{-1/2}\sigma ^{x}\otimes Uw^{(n)}U^{\ast }\bigr )\}%
\bigr\rangle.  \notag
\end{equation}%
In particularly, for any diagonal unitary matrix $U=\{e^{i\varphi
_{j}}\delta _{jk}\}_{j,k=1}^{n}$ with distinct $\varphi
_{j}\in \lbrack 0,2\pi ),\;j=1,...,n$ we obtain%
\begin{equation}
\bigl\langle(\exp \{itH^{(n)}\})_{\alpha j,\beta
k}\bigr\rangle=e^{i(\varphi
_{j}-\varphi _{k})}\bigl\langle(\exp \{itH^{(n)}\})_{\alpha j,\beta k}%
\bigr\rangle.  \notag
\end{equation}%
This implies
\begin{equation}
\bigl\langle U_{\alpha j,\beta k}^{(n)}(t)\bigr\rangle=U_{\alpha
\beta ,j}^{(n)}(t)\delta _{jk},\quad U_{\alpha \beta
,j}^{(n)}(t)=\bigl\langle\exp \{itH^{(n)}\}_{\alpha j,\beta j
}\bigr\rangle.  \label{selr}
\end{equation}%
Hence we can write (\ref{unUn}) as
\begin{equation}
u_{\alpha \gamma }^{(n)}(t)=\frac{1}{n}\sum_{j=1}^{n}U_{\alpha
\gamma ,j}^{(n)}(t).  \label{ualfgam}
\end{equation}%
It follows now from (\ref{D}), (\ref{DF}), (\ref{Hn}), and
(\ref{wkj}) that
\begin{align}
\bigl\langle e^{itH^{(n)}}\bigr\rangle_{\alpha j,\beta j} &=\bigl (%
e^{itH_{0}^{(n)}}\bigr )_{\alpha j,\beta j}+i\int_{0}^{t}ds\;\bigl\langle %
e^{i(t-s)H_{0}^{(n)}}M^{(n)}e^{itH^{(n)}}\bigr\rangle_{\alpha
j,\beta j}\label{DD}
\\
& =e^{it\lambda _{\alpha j}^{(n)}}\delta _{\alpha \beta }+\frac{iv}{\sqrt{n}}%
\int_{0}^{t}ds\;e^{i(t-s)\lambda _{\alpha j}^{(n)}}\sum_{m}\bigl\langle %
w_{jm}U_{-\alpha m,\beta j}^{(n)}(s)\bigr\rangle\notag
\\
=e^{it\lambda _{\alpha j}^{(n)}}\delta _{\alpha \beta }-\frac{v^{2}}{n}%
&\int_{0}^{t}ds\;e^{i(t-s)\lambda _{\alpha
j}^{(n)}}\int_{0}^{s}d\tau \sum_{m}\sum_{\nu }\bigl\langle
U_{-\alpha m,\nu m}^{(n)}(s-\tau )U_{-\nu j,\beta j}^{(n)}(\tau
)\bigr\rangle.  \notag
\end{align}%
Hence taking into account (\ref{selr}) and (\ref{ualfgam}), we
obtain:
\begin{align}
&U_{\alpha \beta ,j}^{(n)}(t)=\int_{0}^{t}ds\;e^{i(t-s)\lambda
_{\alpha j}^{(n)}}\;r_{\alpha \beta ,j}^{(n)}(s)\label{EqUn}
\\
& +e^{it\lambda _{\alpha j}^{(n)}}\delta _{\alpha \beta }
-v^{2}\int_{0}^{t}ds\;e^{i(t-s)\lambda _{\alpha
j}^{(n)}}\int_{0}^{s}d\tau \sum_{\nu }u_{-\alpha \nu
}^{(n)}(s-\tau )U_{-\nu \beta ,j}^{(n)}(\tau ),\notag
\end{align}%
where
\begin{equation}
r_{\alpha \beta ,j}^{(n)}(s)=-\frac{v^{2}}{n}\int_{0}^{s}d\tau
\sum_{m}\sum_{\nu}\bigl\langle\smash{\overset{\circ }{U}}_{-\alpha
m,\nu
m}^{(n)}(s-\tau )U_{-\nu j,\beta j}^{(n)}(\tau )\bigr\rangle,\quad (\overset{%
\circ }{U}=U-\langle U\rangle ).  \notag
\end{equation}%
By using Schwartz inequality and inequality (\ref{v1}) below we
have that
\begin{equation}
\bigl |r_{\alpha \beta ,j}^{(n)}(s)\bigr |\leq
\frac{Cs^{3}}{n^{3/2}}. \label{ro<}
\end{equation}%
Here and below we use the notation $C$ for all positive quantities
that do not depend on $n$, $t$, $z$ and indexes.

\noindent We will use the notations $G_{j}^{(n)}(z)$, $f^{(n)}(z)$
and $R_{j}^{(n)}(z)$
for the generalized Fourier transforms (see (\ref{Fur})) of the $2\times 2$%
-matrices%
\begin{equation}
U_{j}^{(n)}=\{U_{\alpha \beta ,j}^{(n)}(t)\}_{\alpha \beta =\pm
},\; u^{(n)}=\{u_{\alpha \beta }^{(n)}(t)\}_{\alpha \beta =\pm
},\; r_{j}^{(n)}(t)=\{r_{\alpha \beta ,j}^{(n)}(t)\}_{\alpha
\beta =\pm }. \label{Uur}
\end{equation}%
We have from the spectral theorem, (\ref{unUn}), (\ref{selr}),
and (\ref{resHn}) (cf (\ref{Fur}))
\begin{equation}
G_j^{(n)}(z)=i^{-1}\int_0^\infty
e^{-izt}U_j^{(n)}(t)dt=\bigl\{\langle G_{\alpha j,\beta
j(z)}\rangle \bigr\}_{\alpha,\beta=\pm},\quad \Im z<0. \label{Gnj}
\end{equation}
This, (\ref{fag}), and (\ref{EqUn}) lead to the matrix relation:
\begin{equation}
(E_{j}^{(n)}+s\sigma ^{z}-z-v^{2}\sigma ^{x}f^{(n)}(z)\sigma
^{x})G_{j}^{(n)}(z)=1+R_{j}^{(n)}(z).  \notag
\end{equation}%
Since the resolvent $G^{(n)}(z)$ possesses the property $\Im z\Im
(G^{(n)}(z)x,x)\geq 0$, $\forall x\in \mathbb{C}^{2n}$, the matrix $%
f^{(n)}(z)$ possesses the same property $\forall \xi \in
\mathbb{C}^{2}$, and
\begin{equation}
\Im ((E_{j}^{(n)}+s\sigma ^{z}-z-v^{2}\sigma ^{x}f^{(n)}(z)\sigma
^{x})\xi ,\xi )=-\Im z||\xi ||^{2}-v^{2}(f^{(n)}(z)\sigma ^{x}\xi
,\sigma ^{x}\xi )\geq -\Im z||\xi ||^{2}.  \notag
\end{equation}%
Thus the matrix $E_{j}^{(n)}+s\sigma ^{z}-z-v^{2}\sigma
^{x}f^{(n)}(z)\sigma ^{x}$ is invertible, its inverse
\begin{equation}
f^{(n)}(E_{j}^{(n)},z)=(E_{j}^{(n)}+s\sigma ^{z}-z-v^{2}\sigma
^{x}f^{(n)}(z)\sigma ^{x})^{-1}  \label{fnE}
\end{equation}%
admits the bound
\begin{equation}
||f^{(n)}(E_{j}^{(n)},z)||\leq |\Im z|^{-1},  \label{fnj<}
\end{equation}%
and equation for $G_{j}^{(n)}(z)$ takes the form
\begin{equation}
G_{j}^{(n)}(z)=f^{(n)}(E_{j}^{(n)},z)+f^{(n)}(E_{j}^{(n)},z)R_{j}^{(n)}(z).
\label{Gn}
\end{equation}%
Applying to the equation the operation $n^{-1}\sum_j$ we obtain
in view of (\ref{nun})
\begin{equation}
f^{(n)}(z)=\int_{-\infty }^{\infty }\nu
_{0}^{(n)}(dE)f^{(n)}(E,z)+n^{-1}%
\sum_{j=1}^{n}f^{(n)}(E_{j}^{(n)},z)R_{j}^{(n)}(z).  \label{fn}
\end{equation}%
Since the resolvent $G^{(n)}(z)$ is analytic if $\Im z\neq 0$ and
bounded from above by $|\Im z|^{-1}$, we have the bound
\begin{equation}
||f^{(n)}(z)||\leq |\Im z|^{-1},  \label{fn<}
\end{equation}%
implying that the sequence $\{f^{(n)}(z)\}_{n\geq 1}$ consists of
functions, analytic and uniformly bounded in $n$ and in $z$ by
$\eta _{0}^{-1}$ if $|\Im z|\geq \eta _{0}>0$. Hence there exists
analytic $2\times 2$ matrix function $f(z)$, $\Im z\neq 0$,
such that $||f(z)||\leq |\Im z|^{-1}$, and an infinite subsequence $%
\{f^{(n_{k})}(z)\}_{k\geq 1}$ that converges to $f(z)$ uniformly
on any
compact set of $\mathbb{C}\setminus \mathbb{R}$. This and estimates (\ref%
{ro<}), (\ref{fnj<}) allow us to pass to the limit
$n_{k}\rightarrow \infty $ in (\ref{fn}) and obtain that the
limit of any converging subsequence of the sequence
$\{f^{(n)}(z)\}_{n\geq 1}$ satisfies the matrix functional
equation
\begin{equation}
f(z)=\int_{-\infty }^{\infty }\nu _{0}(dE)f(E,z),  \label{feq}
\end{equation}%
where
\begin{equation}
f(E,z)=(E+s\sigma ^{z}-z-v^{2}\sigma ^{x}f(z)\sigma
^{x})^{-1},\quad ||f(E,z)||\leq |\Im z|^{-1}.  \label{fEz<}
\end{equation}%
The equation is uniquely soluble in the class of $2\times 2$ matrix functions, analytic for $%
\Im z\neq 0$, and such that $\Im z\Im (f(z)\xi ,\xi )\geq 0$,
$\forall \xi \in \mathbb{C}^{2}$. Indeed, for any two solutions
$f_{1}$, $f_{2}$ of the class, and $g=f_{1}-f_{2}$ we have
\begin{align*}
g(z)=\int_{-\infty }^{\infty }\nu _{0}(dE)\bigl[E+s\sigma
^{z}-z-&v^{2}\sigma ^{x}f_{1}(z)\sigma ^{x}\bigr]^{-1}v^{2}\sigma
^{x}g(z)\sigma ^{x}
\\
&\times\bigl[E+s\sigma ^{z}-z-v^{2}\sigma ^{x}f_{2}(z)\sigma
^{x}\bigr]^{-1}, \notag
\end{align*}%
and by (\ref{nu}), (\ref{fEz<}) we obtain inequality
$||g(z)||\leq v^{2}|\Im z|^{-2}||g(z)||$ from which it follows
that $g(z)=0$ for $v\Im z<1$, hence for any $\Im z \neq 0$ by
analyticity. The solution of (\ref {feq}) is diagonal, $f_{\alpha
\beta }=f_{\alpha }\delta _{\alpha
\beta }$, and pair $f_{\alpha }$, $\alpha =\pm $ satisfies system (\ref%
{falf}). This follows from the unique solvability of (\ref{falf}).
We can rewrite (\ref{falf}) in the form $f_{\alpha }(z)=f_{\alpha
}^{0}(z+v^{2}f_{-\alpha }(z))$, where $f_{\alpha }^{0}(z)$ is the
Stieltjes transform of the unit non-negative measure $\nu
_{0}^{\alpha }(E)=\nu
_{0}(E-\alpha s)$. Since $f_{\alpha }^{0}(z)$ possesses the property $%
\lim_{\eta \rightarrow \infty }\eta |f_{\alpha }^{0}(i\eta )|=1$
and $|\Im (z+v^{2}f_{-\alpha }(z))|\geq |\Im z|$, then
$\lim_{\eta \rightarrow \infty }\eta |f_{\alpha }(i\eta )|=1$ and
$f_{\alpha }(z)$, $\alpha =\pm $ are Stieltjes transforms of the
unit non-negative measures $\nu _{\alpha }(\lambda )$ (\ref{nal})
(see \cite{Ak-Gl:93}, Section 59).

In addition, the Tchebyshev inequality and bound (\ref{v2}) below imply that for any $%
\varepsilon >0$
\begin{equation}
\mathbf{P}\{|f_{\alpha \gamma }^{(n)}(z)-g_{\alpha \gamma
}^{(n)}(z)|>\varepsilon \}\leq \frac{1}{\varepsilon
^{2}}\mathbf{Var}\{g_{\alpha \gamma }^{(n)}(z)\}\leq
\frac{2v^{2}}{\varepsilon ^{2}n^{2}|\Im z|^{4}}.  \notag
\end{equation}%
Hence the series
\begin{equation*}
\sum_{n=1}^{\infty }\mathbf{P}\{|f_{\alpha \gamma
}^{(n)}(z)-g_{\alpha \gamma }^{(n)}(z)|>\varepsilon \}
\end{equation*}%
converges for any $\varepsilon >0$, $|\Im z|\geq \eta _{0}>0$,
and by the
Borel-Cantelli lemma we have for any fixed z, $|\Im z|\geq \eta _{0}>0$, $%
\lim_{n\rightarrow \infty }g_{\alpha \gamma }^{(n)}(z)=f_{\alpha
\gamma }(z)$ with probability 1. With the same probability this
limiting relation is valid for all points of an infinite
countable sequence $\{z_{j}\}_{j\geq 1},\;|\Im z_{j}|\geq \eta
_{0}>0$, possessing an accumulation point. Hence on any
compact of $\mathbb{C}\setminus \mathbb{R}$ with probability 1 $%
\lim_{n\rightarrow \infty }g_{\alpha \gamma }^{(n)}(z)=f_{\alpha \gamma }(z)$%
, and we have the weak convergence (\ref{nunu}) with the formulas
(\ref{ST})-(\ref{falf}).

Let us prove assertion (ii) of theorem. It follows from the
(\ref{falf})
\begin{align}
\Im f_\alpha(z)&\leq
\underset{\mu\in\mathbb{R}}{\sup}\;\nu_0^{\prime
}(\mu)\int_{-\infty}^\infty\frac{\Im(z+v^2f_{-\alpha}(z))dE}
{(E+s\alpha-\Re(z+v^2f_{-\alpha}(z)))^2+(\Im(z+v^2f_{-\alpha}(z)))^2}\notag
\\
&=\pi\underset{\mu\in\mathbb{R}}{\sup}\;\nu_0^{\prime
}(\mu).\notag
\end{align}
We have now by (\ref{nal})
\begin{align}
\nu_\alpha(\triangle)\leq
|\triangle|\;\underset{\mu\in\mathbb{R}}{\sup}\;\nu_0^{\prime
}(\mu).\notag
\end{align}
This implies (ii). To prove (iii) we notice that by (ii) measures
$\nu_\alpha$, $\alpha=\pm$ are continuous. Thus we can pass to
the limit $n\rightarrow\infty$ in (\ref{mer1}), written as
\begin{align}
\overline{\nu}^{(n)}_\alpha(\lambda)=\nu^{(n)}_\alpha([\lambda+\varepsilon,
\lambda-\varepsilon)). \notag
\end{align}
This and (\ref{mer}) imply (\ref{merl})-(\ref{merl1}).
\end{proof}

\begin{lemma}
\label{l:vars} Under the conditions of Theorem \ref{t:mom1}
\begin{align}
& \mathbf{Var}\{U_{\alpha j,\gamma k}^{(n)}(t)\}\leq
\frac{v^{2}t^{2}}{n},& &\mathbf{Var}\{u_{\alpha \gamma
}^{(n)}(t)\}\leq \frac{v^{2}t^{2}}{n^{2}},
\label{v1} \\
& \mathbf{Var}\{G_{\alpha j,\gamma k}^{(n)}(z)\}\leq
\frac{v^{2}}{n|\Im
z|^{4}},& &\mathbf{Var}\{g_{\alpha \gamma }^{(n)}(z)\}\leq \frac{2v^{2}}{%
n^{2}|\Im z|^{4}},\quad \Im z\neq 0.  \label{v2}
\end{align}
\end{lemma}
\begin{proof}
Acting as in the case of Theorem \ref{t:sa} we obtain (\ref{v1}).
The differentiation formula for the resolvent
\begin{equation}
\frac{\partial G_{\alpha j,\beta k}^{(n)}(z)}{\partial w_{lm}^{(n)}}=\frac{iv%
}{\sqrt{n}}\sum_{\kappa}G_{\alpha j,\kappa l}^{(n)}(z)G_{-\kappa
m,\beta k}^{(n)}(z),  \notag
\end{equation}%
following from the resolvent identity, together with Poincare-Nash
inequality (\ref{PN}) imply the first inequality in (\ref{v2}):
\begin{align}
\mathbf{Var}\{G_{\alpha j,\gamma k}^{(n)}(z)\}&\leq \frac{v^{2}}{n}\bigl \langle%
\sum_{l,m}\bigl |\sum_{\kappa }G_{\alpha j,\kappa
l}^{(n)}(z)G_{-\kappa m,\beta k}^{(n)}(z)\bigr |^{2}\bigr
\rangle  \notag
\\
&\leq \frac{v^{2}}{n}\bigl \langle\sum_{l,\kappa }\bigl |G_{\alpha
j,\kappa
l}^{(n)}(z)\bigr |^{2}\sum_{m,\kappa }\bigl |G_{-\kappa m,\beta k}^{(n)}(z)%
\bigr |^{2}\bigr \rangle\leq \frac{v^{2}}{n|\Im z|^{4}},\quad \Im
z>0. \notag
\end{align}%
The second inequality in (\ref{v2}) can be proved by a similar
argument.
\end{proof}

\medskip

We will return now to functions of variable $t$ and find the $%
n\rightarrow\infty$ limits of the sequences
$\{U^{(n)}_j(t)\}_{n\geq 1}$ and
$\{u^{(n)}(t)=n^{-1}\sum_{j=1}^nU^{(n)}_j(t)\}_{n\geq 1}$.

\begin{theorem}
\label{t:Uu} Consider the $2\times 2$ matrices
$\{U_{j}^{(n)}(t)\}_{n\geq 1}$ and $\{u^{(n)}(t)\}_{n\geq 1}$,
defined in (\ref{Uur}) and (\ref{ualfgam}) and choose a
subsequence $\{E_{j_{n}}^{(n)}\}$ that converges to a given $E$ of
the support of $\nu _{0}$ of (\ref{nu}). Then there exist the
limits
\begin{equation}
U_{E}(t)=\lim_{n\rightarrow \infty }U_{j_{n}}^{(n)}(t),\quad
u(t)=\lim_{n\rightarrow \infty }u^{(n)}(t), \label{limuU}
\end{equation}%
where
\begin{equation}
U_{E}(t)=\frac{i}{2\pi }\int_{L}e^{izt}f(E,z)dz:=\frac{i}{2\pi }%
\lim_{N\rightarrow \infty }\int_{-N-i\eta }^{N-i\eta
}e^{izt}f(E,z)dz,\quad \forall \eta >0,  \label{UE}
\end{equation}%
\begin{equation}
||U_{E}(t)||=1\quad \forall t\geq 0,  \label{|UE|}
\end{equation}%
and
\begin{equation}
u(t)=\frac{i}{2\pi }\int_{L}dz\;e^{izt}\int_{-\infty }^{\infty
}\nu _{0}(dE)f(E,z)  \label{u}
\end{equation}
with $f(E,z)$ defined in (\ref{fEz<}):
\begin{equation}
f_{\alpha\beta}(E,z)=f_{\alpha}(E,z)\delta_{\alpha\beta}, \quad
f_{\alpha}(E,z)=(E_\alpha
-z-v^2f_{-\alpha}(z))^{-1},\;E_\alpha=E+\alpha s.
  \label{faEz}
\end{equation}
\end{theorem}
\begin{proof}
It follows from (\ref{Gnj}), (\ref{Gn}), and the inversion
formula for the generelized Fourier transform \cite{Ti:86} that
\begin{equation}
U_{j_n}^{(n)}(t)=Q^{(n)}(E_{j_n}^{(n)},t)+i^{-1}\int_{0}^{t}Q^{(n)}(E_{j_n}^{(n)},t-s)r
_{j_n}^{(n)}(s)ds,  \label{Uj}
\end{equation}%
where
\begin{align}
&Q^{(n)}(E_{j_n}^{(n)},t) =\frac{i}{2\pi
}\int_{L}e^{izt}f^{(n)}(E_{j_n}^{(n)},z)dz =U_{E}(t)
\label{Qn} \\
&+\frac{i}{2\pi }\int_{L}e^{izt}f^{(n)}(E_{j_n}^{(n)},z)\bigl[%
E-E_{j_n}^{(n)}+v^{2}\sigma ^{x}(f^{(n)}(z)-f(z))\sigma
^{x}\bigr]f(E,z)dz.  \notag
\end{align}%
The resolvent identity yields
\begin{equation}
f^{(n)}(E_{j_n}^{(n)},z)=-\frac{1}{z-E}+\frac{1}{z-E}\bigl[s\sigma
^{z}+(E-E_{j_n}^{(n)})+v^{2}\sigma ^{x}f^{(n)}(z)\sigma
^{x}\bigr]f^{(n)}(E_{j_n}^{(n)},z), \notag
\end{equation}%
and we have for sufficiently big $\eta=|\Im z|$:
\begin{equation}
||f^{(n)}(E_{j_n}^{(n)},z)||\leq \frac{2}{|z-E|},\quad ||f(E,z)||\leq \frac{2}{|z-E|}%
.  \label{fE<}
\end{equation}%
This together with (\ref{fn<}) allow us to pass to the limit under
the integral in the r.h.s. of (\ref{Qn}) and to show that it vanishes as $%
n\rightarrow \infty $. Moreover, we conclude that integral in the r.h.s. of (%
\ref{Qn}) is bounded uniformly in $n$ and $\forall t\geq 0$. The
uniform boundedness of the matrix $U_{E}(t)$ follows from the
equalities
$(U_E)_{\alpha\beta}(t)=(U_E)_{\alpha\alpha}(t)\delta_{\alpha\beta}$
and
\begin{equation}
(U_{E})_{\alpha \alpha }(t)=e^{iE_\alpha t}+\frac{i}{2\pi }%
\int_{L}dz\;e^{itz}\frac{v^{2}f_{-\alpha }(z)}{(E_\alpha
-z-v^{2}f_{-\alpha }(z))(E_\alpha -z)}.  \notag
\end{equation}%
Hence $Q^{(n)}(E_{j_n}^{(n)},t)$ converges to $U_{E}(t)$ as
$n\rightarrow \infty $
and is uniformly bounded in $n$ and $t$. This together with (\ref{Uj}), (\ref%
{ro<}) and equality $||U_{j}^{(n)}(t)||=1$, $\forall t\geq 0$ give us (\ref%
{UE}) and (\ref{|UE|}).

To prove (\ref{u}) notice first that we have from (\ref{fn})
\begin{align}
u^{(n)}(t)&=\frac{i}{2\pi}\int_L
dz\;e^{izt}\int_{-\infty}^\infty\nu^{(n)}_0(dE)f(E,z)  \label{un=}
\\
&+\frac{i}{2\pi}\int_L
dz\;e^{izt}\int_{-\infty}^\infty\nu_0^{(n)}(dE)\bigl[
f^{(n)}(E,z)-f(E,z)\bigr]  \notag
\\
&+i^{-1}\int_0^tn^{-1}\sum_{j=1}^nQ^{(n)}(E_j^{(n)},t-s)r
_j^{(n)}(s)ds. \notag
\end{align}
We integrate by parts with respect to $z$ in the first integral
to obtain in view of (\ref{faEz})
\begin{equation}
-\frac{1}{2\pi t}\int_{-\infty}^\infty\nu^{(n)}_0(dE)\int_L
dz\;e^{izt}\frac{ 1+v^2f^{\prime}_{-\alpha}(z)}{(E_\alpha
-z-v^2f_{-\alpha}(z))^2 }.  \notag
\end{equation}
It follows from (\ref{feq})-(\ref{fEz<}) and (\ref{nu}) that
$||f_{\alpha}(z)||=|z|^{-1}(1+o(1))$, $|z|\rightarrow\infty$ and
$||f^{\prime}_{\alpha}(z)||\leq |\Im z|^{-2}$. Thus the integral
with respect to $z$ is bounded and continuous function of $E$.
This and the weak convergence  $\nu_0^{(n)}$ to $\nu_0$ (see
(\ref{nu})) yield the convergence of the first term of the r.h.s.
of (\ref{un=}) to the r.h.s. of (\ref{u}).

Furthermore, by using (\ref{fE<}), (\ref{fn<}) and (\ref{nun}) we
obtain
\begin{align}
&\int_{-\infty}^\infty\nu^{(n)}_0(dE)
\int_L|dz|\frac{||f^{(n)}(z+E)-f(z+E)||}{|z|^2}  \notag \\
&\leq C\biggl\{\int_{|E|\geq T}\nu^{(n)}_0(dE) + \int_{|x|\geq A}\frac{dx}{%
x^2+\eta^2}+ \underset{|y|\leq A+T}{\max}||f^{(n)}(y-i\eta)-f(y-i\eta)||%
\biggr\}.  \notag
\end{align}
For any $\varepsilon>0$ choosing consequently $A=A(\varepsilon)$, $%
T=T(\varepsilon,A)$, $N_0=N_0(\varepsilon,A,T)$, and taking in account (\ref%
{nu}), and convergence $f^{(n)}(z)$ to $f(z)$ on any compact set
in $\mathbb{C}\setminus\mathbb{R}$, we obtain that the second
term of the r.h.s. of (\ref{un=}) vanishes as
$n\rightarrow\infty$.

At last (\ref{ro<}) yields for the third term of the r.h.s. of
(\ref{un=}):
\begin{align}
&\int_0^tds \frac{Cs^3}{n}\frac{1}{n}\sum_{j}||
\int_L dz\;e^{iz(t-s)}f^{(n)}(E_j,z)||  \notag \\
&\leq\int_0^tds \frac{Cs^3}{n} \int_{-\infty}^\infty\nu^{(n)}_0(dE)\biggl\{%
||U_E(t-s)||+\int_Ldz||f^{(n)}(E,z)-f(E,z)||\biggr\},  \notag
\end{align}
and taking into account (\ref{|UE|}), (\ref{fE<}), and (\ref{fn<})
we conclude that the term also vanishes as $ n\rightarrow\infty$
uniformly in $t$, varying on a compact interval.
\end{proof}

\section{Time Evolution}

\label{s:evol}

\medskip
We will prove now the main general result of \cite{Le-Pa:04}, a
formula for the limit as $ n\rightarrow\infty$ of the expectation
(\ref{ron}) of the reduced density matrix (\ref{rdmo}) of our
model formula (4.7) of \cite{Le-Pa:04}.

\noindent
\begin{theorem}\label{t:rdm} Consider the model of composite
system, defined by (\ref{nun})-(\ref{mun}). Choose a subsequence
$\{E_{k_{n}}^{(n)}\}$ of eigenvalues of $h_n$ of (\ref{Ham}) that
converges to a certain $E\in\text{supp}\;\nu _{0}$. Then we have
for the limit as $n\rightarrow \infty $ of the expectation
(\ref{ron}) of the reduced density matrix (\ref{rdmo}) uniformly
in $t$ varying on a finite interval:
\begin{align}
&\rho _{\alpha ,\delta }(E,t):=\lim_{n\rightarrow \infty }\rho
_{\alpha ,\delta
}^{(n)}(E^{(n)}_{k_n},t)=\frac{1}{(2\pi)^{2}}\int_{L_{2}}dz_{2}
\int_{L_{1}}dz_{1}e^{it(z_{2}-z_{1})}\label{rrr}
\\
& \times\frac{f_{\alpha }(E,z_{1})f_{\delta }(E,z_{2})\rho
_{\alpha ,\delta }(0)+v^{2}f_{-\alpha }(E,z_{1})f_{-\delta
}(E,z_{2})f_{\alpha ,\delta }(z_{1},z_{2})\rho _{-\alpha ,-\delta
}(0)}{1-v^{4}f_{\alpha ,\delta }(z_{1},z_{2})f_{-\alpha ,-\delta
}(z_{1},z_{2})},  \notag
\end{align}
where $L_{1}=(-\infty+i\eta_{1},\infty+i\eta_{1})$, $L_{2}=(-\infty-i%
\eta_{2},\infty-i\eta_{2})$, $\eta_{1}>0$, $\eta_{2}>0$;
\begin{equation}
f_{\beta,\gamma }(z_{1,}z_{2})=\int_{-\infty}^\infty \nu
_{0}(dE)f_{\beta }(E,z_{1})f_{\gamma }(E,z_{2})  \label{f2}
\end{equation}
with $f_\alpha(E,z)$ defined in (\ref{faEz}).
\end{theorem}
\begin{proof}
In view of (\ref{ron}) it suffices to prove the following
expression for the average transfer matrix (\ref{per}):
\begin{align}
T_{\alpha \beta \gamma \delta }(E,t) &:= \lim_{n\rightarrow \infty
}T^{(n)}_{\alpha \beta \gamma \delta }(E^{(n)}_{k_n},t)\label{T}
\\
&=\frac{1}{(2\pi)^{2}}
\int_{L_{2}}dz_{2}\int_{L_{1}}dz_{1}e^{it(z_{2}-z_{1})}
\widetilde{T}_{ \alpha \beta \gamma \delta }(E,z_{1},z_{2}),\notag
\end{align}
where  the "two-point" functions $\widetilde{T}_{ \alpha \beta
\gamma \delta }(E,z_{1},z_{2})$ are analytic in $z_{1}$ and in
$z_{2}$ outside the real axis and have the form
\begin{align}
&\widetilde{T}_{ \alpha \beta \gamma \delta
}(E,z_{1},z_{2})=f_{\beta }(E,z_{1})f_{\gamma }(E,z_{2})\label{FT}
\\
\times&(\delta_{\alpha ,\beta}\delta _{\gamma ,\delta
}+v^{2}\delta _{-\alpha ,\beta}\delta _{-\gamma ,\delta }f_{-\beta
,-\gamma }(z_{1},z_{2}))[1-v^{4}f_{ \beta,\gamma
}(z_{1},z_{2})f_{-\beta ,-\gamma }(z_{1},z_{2})]^{-1}.\notag
\end{align}
Acting as in the proof of the Theorem \ref{t:mom1}, i.e., by using
 the Duhamel formula (\ref{D}), (\ref{wkj}) and
differentiation formula (\ref{DF}) we obtain the relation (cf
(\ref{DD}))
\begin{align}
T^{(n)}_{\alpha \beta \gamma \delta
}&(E^{(n)}_{k_n},t_1,t_2):=\sum_{j=1}^{n}\left\langle U_{\alpha
j,\beta k}^{(n)}(-t_1)U_{\gamma k,\delta
j}^{(n)}(t_2)\right\rangle=e^{it_2\lambda_{\gamma
k_n}^{(n)}}\delta_{\gamma\delta}U_{\alpha\beta,k_n}^{(n)}(-t_1)
\label{EqT}
\\
&- v^2\int_0^{t_2}ds\;e^{i(t_2-s)\lambda_{\gamma k_n
}^{(n)}}\int_0^sd\tau\sum_{\kappa}T^{(n)}_{\alpha, \beta, -\kappa,
\delta
}(E^{(n)}_{k_n},t_1,\tau)u_{-\gamma\kappa}^{(n)}(s-\tau)  \notag \\
&+v^2\int_0^{t_2}ds\;e^{i(t_2-s)\lambda_{\gamma k_n
}^{(n)}}\int_0^{t_1}d\tau\sum_{\kappa}U_{-\kappa\beta,k_n}^{(n)}(-%
\tau)K^{(n)}_{\alpha, \kappa, -\gamma, \delta}(t_1-\tau,s)  \notag \\
&+\int_0^{t_2}ds\;e^{i(t_2-s)\lambda_{\gamma k_n
}^{(n)}}r_{\alpha\beta\gamma\delta}^{(n)}(t_1,s),\notag
\end{align}
where
\begin{equation}
K^{(n)}_{\alpha \beta \gamma \delta}(t_1,t_2)
=\frac{1}{n}%
\sum_{m=1}^{n}\sum_{j=1}^{n}\left\langle  U^{(n)}_{\alpha j,\beta
m}(-t_1)U^{(n)}_{\gamma m,\delta j}(t_2) \right\rangle, \quad
|K^{(n)}_{\alpha \beta \gamma \delta}(t_1,t_2)|\leq 1, \label{K<}
\end{equation}
and
\begin{align}
&r_{\alpha\beta\gamma\delta}^{(n)}(t_1,t_2)\label{rn}
\\
&=\frac{v^2}{n}\sum_{j} \biggl[-\int_0^{t_2}d\tau
\sum_{m}\sum_{\kappa}\bigl\langle
\smash{\overset{\circ}{U}}^{(n)}_{-\gamma m, \kappa m
}(t_2-\tau)U^{(n)}_{-\kappa k_n,\delta
j}(\tau)U^{(n)}_{\alpha j,\beta k_n}(-t_1)\bigr\rangle  \notag \\
&+\int_0^{t_1}d\tau \sum_{m}\sum_{\kappa}\bigl\langle
\smash{\overset{\circ} {U}}^{(n)}_{-\kappa k_n,\beta k_n
}(-\tau)U^{(n)}_{-\gamma m,\delta j}(t_2)U^{(n)}_{\alpha j,\kappa
m}(\tau-t_1)\bigr\rangle\biggr].\notag
\end{align}
It follows from  Schwartz and Poincare-Nash inequalities and
estimate (\ref{v1}) that (cf (\ref{ro<}))
\begin{equation}
|r_{\alpha\beta\gamma\delta}^{(n)}(t_1,t_2)|=O(n^{-1/2}).
\label{r<}
\end{equation}
Let $\widetilde{T}_{\alpha \beta \gamma \delta
}^{(n)}(E,z_{1},z_{2})$, $\Im z_1>0$, $\Im z_2<0$ be generalized Fourier transform of
$T_{\alpha \beta \gamma \delta }^{(n)}(E,t_{1},t_{2})$:
\begin{equation}
\widetilde{T}_{\alpha \beta \gamma \delta
}^{(n)}(E,z_{1},z_{2})=i\int_{0}^{\infty
}dt_{1}\;e^{iz_{1}t_{1}}\biggl(\frac{1%
}{i}\int_{0}^{\infty
}dt_{2}\;e^{-iz_{2}t_{2}}T_{\alpha \beta \gamma \delta }^{(n)}(E,t_{1},t_{2})%
\biggr),  \notag
\end{equation}%
so that
\begin{equation*}
T_{\alpha \beta \gamma \delta }^{(n)}(E,t_{1},t_{2})=\frac{1}{(2\pi )^{2}%
}\int_{L_{2}}dz_{2}\int_{L_{1}}dz_{1}e^{i(t_{2}z_{2}-t_{1}z_{1})}\widetilde{T%
}_{\alpha \beta \gamma \delta }^{(n)}(E,z_{1},z_{2}),
\end{equation*}%
where $L_{1}=(-\infty +i\eta _{1},\infty +i\eta _{1})$,
$L_{2}=(-\infty -i\eta _{2},\infty -i\eta _{2})$, $\eta _{1}>0$,
$\eta _{2}>0$. In view of relations
\begin{equation*}
U_{\alpha j,\beta k}^{(n)}(-t_{1})=\overline{U_{\alpha k,\beta
j}^{(n)}(t_{1})}, \quad G_{\alpha j,\beta
k}^{(n)}(z_{1})=\overline{G_{\alpha k,\beta
j}^{(n)}(\overline{z_{1}})}
\end{equation*}%
we have
\begin{equation*}
i\int_{0}^{\infty }dt_{1}\;e^{iz_{1}t_{1}}U_{\alpha j,\beta
k}^{(n)}(-t_{1})=G_{\alpha j,\beta k}^{(n)}(z_{1})
\end{equation*}%
and (\ref{EqT}) yields
\begin{align}
\widetilde{T}_{\alpha \beta \gamma \delta
}^{(n)}&(E_{k_n}^{(n)},z_{1},z_{2})\label{TTeq}
\\
&=\frac{1}{\lambda _{\gamma k_n}^{(n)}-z_{2}}\biggl[ G_{\alpha
\beta ,k_n}^{(n)}(z_{1})\delta _{\gamma \delta }+v^{2}\sum_{\kappa
}f_{-\gamma \kappa }^{(n)}(z_{2})\widetilde{T}_{\alpha \beta
-\kappa \delta }^{(n)}(E_{k_n}^{(n)},z_{1},z_{2})\notag
\\
& +v^{2}\sum_{\kappa}G_{-\kappa \beta ,k_n}^{(n)}(z_{1})\widetilde{K}%
_{\alpha, \kappa, -\gamma, \delta
}^{(n)}(z_{1},z_{2})+\widetilde{r}_{\alpha \beta \gamma \delta
}^{(n)}(z_{1},z_{2})\biggr].\notag
\end{align}%
Here $\widetilde{K}^{(n)}$ and $\widetilde{r}^{(n)}$ are
generalized Fourier transforms of ${K}^{(n)}$ and ${r}^{(n)}$ of
(\ref{K<}) and (\ref{rn}) respectively, and as it follows from
(\ref{K<}) the absolute values of $\widetilde{K}_{\alpha \beta
\gamma \delta }^{(n)}(z_{1},z_{2})$ are bounded uniformly in $n$
by $|\Im z_{1}|^{-1}|\Im z_{2}|^{-1}$.

To write (\ref{TTeq}) in the matrix form for any fixed pair $%
\alpha $, $\beta $ we denote $\widetilde{K}_{\alpha \beta }^{(n)}$, $%
\widetilde{S}_{\alpha \beta }^{(n)}$, $\widetilde{r}_{\alpha
\beta }^{(n)}$ the $2\times 2$-matrices, which entries are
$(\widetilde{K}_{\alpha \beta }^{(n)})_{\gamma \delta
}=\widetilde{K}_{\alpha \beta \gamma \delta }^{(n)}$ etc., and
$\widetilde{K}_{\alpha \beta }^{-(n)}$ are $2\times 2$-matrices
with the entries $(\widetilde{K}_{\alpha \beta }^{-(n)})_{\gamma \delta }=%
\widetilde{K}_{\alpha, \beta, -\gamma, \delta }^{(n)}$, $\gamma
,\delta =\pm $, so
\begin{align*}
\widetilde{T}_{\alpha
\beta}^{(n)}(E_{k_n}^{(n)},z_{1},z_{2})=f^{(n)}&(E_{k_n}^{(n)},z_{2})\widetilde{r}_{\alpha
\beta }^{(n)}(z_{1},z_{2})+f^{(n)}(E_{k_n}^{(n)},z_{2})
 \\
& \times \biggl[G_{\alpha \beta
,k_n}^{(n)}(z_{1})\mathbf{1}_{2}+v^{2}\sum_{\kappa }G_{-\kappa
\beta ,k_n}^{(n)}(z_{1})\widetilde{K}_{\alpha \kappa
}^{-(n)}(z_{1},z_{2})\biggr].
\end{align*}%
Plugging expression (\ref{Gn}) for $G_{k_n}^{(n)}(z_{2})$  we
obtain
\begin{align}
&\widetilde{T}_{\alpha \beta }^{(n)}(E_{k_n}^{(n)},z_{1},z_{2})
=\widetilde{R}_{\alpha \beta
}^{(n)}(E_{k_n}^{(n)},z_{1},z_{2})\label{TMeq}
\\
&+f^{(n)}(E_{k_n}^{(n)},z_{2})
\biggl[f_{\alpha \beta }^{(n)}(E_{k_n}^{(n)},z_{1})\mathbf{1}_{2}+v^{2}\sum_{%
\kappa }f_{-\kappa \beta
}^{(n)}(E_{k_n}^{(n)},z_{1})\widetilde{K}_{\alpha \kappa
}^{-(n)}(z_{1},z_{2})\biggr],\notag
\end{align}%
where reminder $\widetilde{R}_{\alpha \beta }^{(n)}(E_{k_n}^{(n)},z_{1},z_{2})$ is a $%
2\times 2$-matrix, and according to (\ref{ro<}), (\ref{r<}), and
uniform boundedness of $\widetilde{K}^{(n)}(z_{1},z_{2})$ and
$f^{(n)}(E_{m}^{(n)},z)$, we have
\begin{equation}
\lim_{n\rightarrow\infty}||\widetilde{R}_{\alpha \beta
}^{(n)}(E_{k_n}^{(n)},z_{1},z_{2})||=0,\quad
\lim_{n\rightarrow\infty}||n^{-1}\sum_{m=1}^{n}\widetilde{R}_{\alpha
\beta }^{(n)}(E_{m}^{(n)},z_{1},z_{2})||=0.\label{R<}
\end{equation}%
Applying the operation $n^{-1}\sum_{m=1}^{n}$ to (\ref{TMeq})
with $k_n=m$ we obtain:
\begin{align}
&\widetilde{K}_{\alpha \beta
}^{(n)}(z_{1},z_{2})=v^{2}\sum_{\kappa}\int_{-\infty }^{\infty
}\nu
_{0}^{(n)}(dE)f_{-\kappa \beta }^{(n)}(E,z_{1})f^{(n)}(E,z_{2})\widetilde{K}%
_{\alpha \kappa }^{-(n)}(z_{1},z_{2})\label{KMn}
\\
&\quad\quad+\int_{-\infty }^{\infty }\nu _{0}^{(n)}(dE)f_{\alpha
\beta }^{(n)}(E,z_{1})f^{(n)}(E,z_{2})+ n^{-1}\sum_{m}
\widetilde{R}_{\alpha \beta }^{(n)}(E_{m}^{(n)},z_{1},z_{2}).
\notag
\end{align}%
This implies that for any fixed $\alpha $, $\beta $, $\gamma $, $%
\delta $ the limiting values $\widetilde{K}_{\alpha \beta \gamma
\delta }(z_{1},z_{2})=\lim_{n\rightarrow 0}\widetilde{K}_{\alpha
\beta \gamma \delta }^{(n)}(z_{1},z_{2})$ and
$\widetilde{K}_{\alpha, -\beta, -\gamma, \delta }(z_{1},z_{2})$
satisfy the system of linear equations
\begin{equation*}
\widetilde{K}_{\alpha \beta \gamma \delta }(z_{1},z_{2})=f_{\beta
\gamma
}(z_{1},z_{2})\bigl[\delta _{\alpha \beta }\delta _{\gamma \delta }+v^{2}%
\widetilde{K}_{\alpha, -\beta, -\gamma, \delta
}(z_{1},z_{2})\bigl],
\end{equation*}%
where $f_{\beta \gamma }(z_{1},z_{2})$ are defined in (\ref{f2}).
Solving this system we obtain
\begin{equation}
\widetilde{K}_{\alpha, -\beta, -\gamma, \delta
}(z_{1},z_{2})=\frac{f_{-\beta, -\gamma }(z_{1},z_{2})\bigl[\delta
_{\alpha, -\beta }\delta _{-\gamma, \delta }+v^{2}f_{\beta, \gamma
}(z_{1},z_{2})\delta _{\alpha, \beta }\delta _{\gamma, \delta
}\bigr]}{1-v^{4}f_{\beta, \gamma }(z_{1},z_{2})f_{-\beta, -\gamma
}(z_{1},z_{2})}.\label{K}
\end{equation}%
Now we return to the variables $t_{1}$, $t_{2}$. It follows from (\ref{TMeq}) and (\ref%
{R<})  that
\begin{align}
T_{\alpha \beta \gamma \delta }(E,t_{1},t_{2})=\lim_{n\rightarrow
\infty }\frac{1}{(2\pi )^{2}}
\int_{L_{2}}dz_{2}\int_{L_{1}}&dz_{1}\;e^{i(t_{2}z_{2}-t_{1}z_{1})}
\label{Teq}
\\
&\times\biggl[f_{\alpha \beta
}^{(n)}(E_{k_n}^{(n)},z_{1})f_{\gamma \delta
}^{(n)}(E_{k_n}^{(n)},z_{2})\notag
\\
+v^{2}\sum_{\kappa,\nu }f_{-\kappa \beta
}^{(n)}(E_{k_n}^{(n)},z_{1})&f_{\gamma \nu
}^{(n)}(E_{k_n}^{(n)},z_{2})\widetilde{K}_{\alpha, \kappa, -\nu,
\delta }^{(n)}(z_{1},z_{2})\biggr],\notag
\end{align}%
and we have to prove the equality:
\begin{align}
T_{\alpha \beta \gamma \delta }(E,t_{1},t_{2})=\frac{1}{(2\pi )^{2}}%
\int_{L_{2}}dz_{2}\int_{L_{1}}dz_{1}&e^{i(t_{2}z_{2}-t_{1}z_{1})}
f_{\beta }(E,z_{1})f_{\gamma }(E,z_{2}) \label{Ttt}
\\
& \times \bigl[\delta _{\alpha \beta }\delta _{\gamma \delta }+v^{2}%
\widetilde{K}_{\alpha, -\beta, -\gamma, \delta
}(z_{1},z_{2})\bigr], \notag
\end{align}%
which together with (\ref{K}) yields (\ref{T}). Notice that for
any fixed non-real $z_{1}$, $z_{2}$ the integrand of (\ref{Teq})
tends to integrand of (\ref{Ttt}), but it has no an integrable
majorant. Because of this fact we replace $\widetilde{K}_{\alpha,
\kappa, -\nu, \delta }^{(n)}(z_{1},z_{2})$
in (\ref{Teq}) by the corresponding entry of the r.h.s. matrix of (\ref%
{KMn}) to obtain
\begin{align*}
&T_{\alpha \beta \gamma \delta }(E,t_{1},t_{2})=(U_{E})_{\beta
\beta }(-t_{1})(U_{E})_{\gamma \gamma }(t_{2})
\delta _{\alpha \beta }\delta _{\gamma \delta }
\\
& +\lim_{n\rightarrow \infty }\frac{v^{2}}{(2\pi
)^{2}}\sum_{\kappa,\nu}\int
\nu_{0}^{(n)}(d\mu)dz_{2}dz_{1}e^{i(t_{2}z_{2}-t_{1}z_{1})}f_{-\kappa
\beta }^{(n)}(E_{k_n}^{(n)},z_{1})f_{\gamma \nu
}^{(n)}(E_{k_n}^{(n)},z_{2})
\\
& \times \bigl[f_{\alpha \kappa}^{(n)}(\mu ,z_{1})f_{-\nu \delta }^{(n)}(\mu ,z_{2})
+v^{2}\sum_{\kappa _{1},\nu _{1}}f_{-\kappa _{1}\kappa }^{(n)}(\mu ,z_{1})
f_{-\nu \nu _{1}}^{(n)}(\mu ,z_{2})%
\widetilde{K}_{\alpha, \kappa _{1},-\nu _{1},\delta
}(z_{1},z_{2})\bigr].
\end{align*}%
Here we denote $\int\nu_{0}^{(n)}(d\mu)dz_{2}dz_{1}=\int_{-\infty
}^{\infty }\nu_{0}^{(n)}(d\mu)\int_{L_2}dz_{2}\int_{L_1}dz_{1}$.
Now it remains to prove that the following expressions
\begin{align}
&\int\nu _{0}^{(n)}(d\mu)dz_{2}dz_{1}\bigl[f^{(n)}(E_{k_n}^{(n)},z_{1})-f(E,z_{1})%
\bigr]f^{(n)}(E_{k_n}^{(n)},z_{2})f^{(n)}(\mu ,z_{1})f^{(n)}(\mu ,z_{2}),  \notag \\
&\int\nu
_{0}^{(n)}(d\mu)dz_{2}dz_{1}f(E,z_{1})f(E,z_{2})\bigl[f^{(n)}(\mu
,z_{1})-f(\mu ,z_{1})\bigr]f^{(n)}(\mu ,z_{2}),  \notag \\
&\int\nu _{0}^{(n)}(d\mu)dz_{2}dz_{1}f(E,z_{1})f(E,z_{2})f(\mu
,z_{1})f(\mu
,z_{2})\bigl[K^{(n)}(z_{1},z_{2})-K(z_{1},z_{2})\bigr],  \notag
\\
&\int_{-\infty }^{\infty }\bigl[\nu _{0}^{(n)}-\nu _{0}\bigr]%
(d\mu
)\int_{L_{2}}dz_{2}\int_{L_{1}}dz_{1}f(E,z_{1})f(E,z_{2})f(\mu
,z_{1})f(\mu ,z_{2})K(z_{1},z_{2}),  \notag
\end{align}%
where $K^{(n)}(z_{1},z_{2})=K_{\alpha \beta \gamma \delta
}^{(n)}(z_{1},z_{2})$, vanish as $n\rightarrow \infty $.
\medskip

Since $||f^{(n)}(\mu,z_1)||\leq \eta_1^{-1}$ and there exist $%
\eta_i$, $i=1,2$ such that $||f^{(n)}(\mu,z_2)||\leq 2|z_2-\mu|^{-1}$, $%
||f^{(n)}(E_{k_n}^{(n)},z_i)||\leq 2|z_i|^{-1}$, $|\Im
z_i|\geq\eta_i$, then the norm of the first expression is bounded
by
\begin{align*}
\frac{1}{\eta_1}\int_{L_1}|dz_1|\frac{1}{|z_1|^2}\bigl(%
v^2||f^{(n)}(z_1)-f(z_1)||+|E_{k_n}^{(n)}-E|\bigr) \int_{-\infty}^\infty\nu_0^{(n)}(d%
\mu)\int_{L_2} \frac{|dz_2|}{|z_2||z_2-\mu|}.
\end{align*}
We have by Schwartz inequality
\begin{align*}
\int_{L_2}
\frac{|dz_2|}{|z_2||z_2-\mu|}\leq\biggl(\int_{-\infty}^\infty\frac{dx}{x^2+\eta_2^2}
\int_{-\infty}^\infty\frac{dx}{(x-\mu)^2+\eta_2^2}\biggr)^{1/2}=\frac{\pi}{\eta_2}.
\end{align*}
This and the uniform convergence of $f^{(n)}$ to $f$ on a compact
set of $\mathbb{C}\setminus\mathbb{R}$ imply that the first
expression vanishes as $n\rightarrow\infty$. Treating similarly
the remaining three expressions we prove that they tends to zero
as $n\rightarrow\infty$.
\end{proof}

\section{Van-Hove Limit}

\label{s:vH}

\noindent In this section we study the limiting case, where the
coupling constant $v$ of the system-reservoir interaction tends
to zero, the time $t$ tends to infinity while the transition
rate, given by first order perturbation in the interaction, is
kept fixed \cite{Da,Ha,Ku-Co,Sp:80}. In terms of (\ref{T}) this
corresponds to making simultaneously the limits
\begin{equation}
v\rightarrow 0,\quad t\rightarrow \infty ,\quad \tau =tv^{2}\quad
\mathrm{\ fixed}  \label{cvH}
\end{equation}%
after the limit $n\rightarrow \infty $, i.e., in formula
(\ref{rrr}).

We note that this limit as well as several other important topics
of the small system-reservoir dynamics were considered by N.N.
Bogolubov in 1945 \cite{Bo:45} in the context of classical
oscillator interacting linearly with the oscillator reservoir.

\begin{theorem}\label{t:vHl}  Let the Fourier transform $\widehat{\nu_0}(u)$ of the
density $\nu^\prime_0$ of the measure $\nu _0$ in (\ref{nu}) be absolutely integrable function:
\begin{equation}
\int _{-\infty}^\infty|\widehat{\nu_0}(u)| du = c_0 < \infty,
\label{c0}
\end{equation}
\begin{equation}
\widehat{\nu_0}(u)= \int _{-\infty}^\infty e^{-iuE}\nu^\prime_0(E)dE.
\label{ft}
\end{equation}
Then the diagonal entries of the limiting reduced density matrix
in (\ref{rrr}) in the van Hove limit are
\begin{align}
&\rho _{\alpha ,\alpha }^{vH}(E,\tau )=
2\pi\biggl[\frac{\nu^\prime_0(E)}{\Gamma _{\alpha }(E)}\rho _{\alpha
,\alpha }(0)+\frac{\nu^\prime_0(E-2\alpha s)}{\Gamma _{-\alpha }(E)}\rho
_{-\alpha ,-\alpha }(0)  \label{rvHf}
\\
&+e^{-\tau \Gamma_\alpha(E)}\frac{\nu^\prime_0(E+2\alpha s)}{\Gamma _{\alpha }(E)}%
\rho _{\alpha ,\alpha }(0)-e^{-\tau \Gamma_{-\alpha}(E)}\frac{%
\nu^\prime_0(E-2\alpha s)}{\Gamma _{-\alpha }(E)}\rho _{-\alpha ,-\alpha
}(0)\biggr] \notag
\end{align}
where
\begin{equation}
\Gamma _{\alpha }(E)=2\pi \left[ \nu^\prime _{0}(E)+\nu^\prime_{0}(E+2\alpha s
)\right]; \label{GvHa}
\end{equation}
and the off-diagonal entries are
\begin{align}
\rho_{\alpha ,-\alpha}^{vH}(E,\tau )=\rho _{\alpha
,-\alpha}(0)e^{-2\alpha sit}e^{i\tau(f_0(E+2\alpha s+i0)-f_0(E-2\alpha s-i0))},\label{off}
\end{align}
where $f_0$ is the Stiltjes transform of $\nu^\prime_0$:
\begin{equation*}
f_0(z)=\int_{-\infty }^{\infty }\frac{\nu^\prime_0(E)dE}
{E-z},\quad\Im z\neq 0.
\end{equation*}%
\end{theorem}

\begin{lemma}\label{l:nuc} In conditions (\ref{c0}), (\ref{ft}) of the
Theorem \ref{t:vHl} next statements for the functions
$f_{\alpha}(z)$, $\alpha=\pm$ are valid:

\begin{enumerate}
\item[(i)] $\sup_{\Im z\geq 0} |f_\alpha(z)|\leq c_0$, \label{sup}

\item[(ii)] $\lim_{v\rightarrow 0}\frac{1}{\pi}\Im f_\alpha(\lambda+i0)=\nu^\prime_0 (\lambda
- \alpha s)$, $\lambda \in \mathbb{R}$.
\end{enumerate}
\end{lemma}
\begin{proof} Estimate (i) follows from the representation of the
functions $f_{\alpha}(z)$, $\Im z >0$ in the form
\begin{equation}
f_{\alpha}(z)=i\int_0^\infty e^{i u (-\alpha s +z+v^2f_{-\alpha}(z))}\widehat{%
\nu_0}(u)du  \label{fa}
\end{equation}
and condition $\Im z\Im f_{\alpha}(z)\geq0$. It also follows from
(\ref{c0}), (\ref{fa}) that
\begin{equation}
\lim_{v\rightarrow 0}f_\alpha(z)=f_0(z-\alpha s)=i\int_0^\infty
e^{i u (z-\alpha s)}\widehat{\nu_0}(u)du,\quad \Im z\geq 0.
\label{faf0}
\end{equation}
Hence
\begin{equation}
\lim_{v\rightarrow 0}\frac{1}{\pi}\Im f_\alpha(\lambda+i0)=
\frac{1}{\pi}\Re \int_0^\infty e^{i u (-\alpha s+\lambda )}\widehat{\nu_0}%
(u)du=\nu^\prime_0(\lambda-\alpha s).
\end{equation}
\end{proof}
\medskip
\begin{proof} (of the Theorem \ref{t:vHl}). By using equalities (see
(\ref{f2}))
\begin{align*}
f_{\alpha,\alpha}(z_1,z_2)=\frac{\delta f_\alpha}{\delta z + v^2
\delta f_{-\alpha}},\quad \delta z =z_1-z_2,\quad \delta
f_{\alpha} =f_{\alpha}(z_1)-f_{\alpha}(z_2),
\end{align*}
and by using analyticity of the integrand of (\ref{rrr}) in $z_1$ and in $%
z_2 $, we can write the following representation for the diagonal
entries of the limiting reduced density matrix
\begin{align}
&\rho _{\alpha ,\alpha }(E,t) =\frac{\rho _{\alpha ,\alpha
}(0)}{(2\pi)^{2}}
\int_{L_{2}^v}dz_{2}\int_{L_{1}^v}dz_{1}e^{-it\delta z}f_\alpha(E,
z_1)f_\alpha(E, z_2) \label{I123}
\\
&+\frac{\rho _{\alpha ,\alpha }(0)}{(2\pi)^{2}}
\int_{L_{2}^v}dz_{2} \int_{L_{1}^v}dz_{1}e^{-it\delta z}
\frac{v^4\delta f_{\alpha}\delta f_{-\alpha}f_\alpha(E,
z_1)f_\alpha(E, z_2)} {\delta z(\delta z+v^2\delta
f_{\alpha}+v^2\delta f_{-\alpha})}\notag
 \\
&+\frac{\rho _{-\alpha ,-\alpha }(0)}{(2\pi)^{2}}
\int_{L_{2}^v}dz_{2} \int_{L_{1}^v}dz_{1}e^{-it\delta z} \frac
{v^2\delta f_{\alpha}(\delta z+v^2\delta
f_{\alpha})f_{-\alpha}(E, z_1)f_{-\alpha}(E, z_2)} {\delta
z(\delta z+v^2\delta f_{\alpha}+v^2\delta f_{-\alpha})}\notag
\\
&=I_1^v(E,t)+I_2^v(E,t)+I_3^v(E,t),\notag
\end{align}
where $L_1^v=\{z_1: \Im z_1= v^2\eta_1\}$, $L_2^v=\{z_2: \Im z_2=-v^2\eta_2\}$, $%
\eta_1$ and $\eta_2$ are arbitrarily chosen positive constants.

To compute the limit (\ref{cvH}) of $I_1^v(E,t)$ we change
variables to $\zeta_j=v^{-2}(z_j-E_\alpha)$, $j=1,2$,
and by Lemma \ref{l:nuc} we have
\begin{align*}
&I_1^v(E,t)=\frac{\rho _{\alpha ,\alpha }(0)}{(2\pi)^{2}} \int_{L_{1}}d%
\zeta_{1}\frac{e^{-i\tau
\zeta_{1}}}{\zeta_{1}+f_{-\alpha}(E_\alpha+v^2
\zeta_{1})} \int_{L_{2}}d\zeta_{2}\frac{e^{i\tau \zeta_{2}}}{%
\zeta_2+f_{-\alpha}(E_\alpha+v^2 \zeta_2)}
\\
&=\frac{\rho _{\alpha ,\alpha }(0)}{%
(2\pi)^{2}} \int_{L_{1}}d\zeta_{1}\frac{e^{-i\tau \zeta_{1}}}{%
\zeta_{1}+f_0(E+2\alpha s+i0)}
\int_{L_{2}}d\zeta_{2}\frac{e^{i\tau
\zeta_{2}}}{\zeta_2+f_0(E+2\alpha s-i0)}+o(1).
\end{align*}
Computing last integrals by residues and applying equality
\begin{equation}
f_0(\lambda +i0)-f_0(\lambda -i0)=2\pi i\nu^\prime_0(\lambda) \label{[f]}
\end{equation}
we obtain
\begin{equation}
\mathrm{vH-}\lim I_1^v(E,t)=\rho _{\alpha ,\alpha }(0)e^{-2\pi
\tau \nu^\prime_0(E+2\alpha s)}, \label{I1}
\end{equation}
where the symbol "vH-lim" denotes the double limit (\ref{cvH}).

Changing variables in $I_3^v(E,t)$ to $
\zeta_2=v^{-2}(z_2-E_{-\alpha})\in L_2=\{\zeta:\Im\zeta=-\eta_2\}$,
$\zeta_1=v^{-2}(z_1-z_2)\in
L_1=\{\zeta:\Im\zeta=\eta_1+\eta_2\}$ yields
\begin{align}
I_3^v(E,t)&=\frac{\rho _{-\alpha ,-\alpha }(0)}{(2\pi)^{2}}
\int_{L_2}d\zeta_2 \int_{L_1}d\zeta_1e^{-i\tau \zeta_1}
\frac{\delta f_\alpha(\zeta_1+\delta
f_\alpha)}{\zeta_1(\zeta_1+\delta f_\alpha+\delta
_{-\alpha})}\label{Iv3}
\\
&\times\frac{1}{\zeta_2+f_\alpha(E_{-\alpha}+v^2\zeta_2)}\cdot
\frac{1}{\zeta_1+\zeta_2+f_\alpha(E_{-\alpha}+v^2(\zeta_1+\zeta_2))}.
\notag
\end{align}
It follows from (\ref{c0}) that the absolute value of integrand
of (\ref{I3}) is bounded from above by
\begin{equation}
\frac{c}{|\zeta_1||\zeta_2||\zeta_1+\zeta_2|}= \frac{c} {\sqrt{%
\lambda_1^2+(\eta_1+\eta_2)^2}\sqrt{\lambda_2^2+\eta_2^2} \sqrt{%
(\lambda_1+\lambda_2)^2+\eta_1^2}},  \notag
\end{equation}
where $c>0$ does not depend on $v$, $\lambda_j=\Re \zeta_j$,
$j=1,2$. Now Schwartz inequality yields for any $B>0$
\begin{align*}
&\int_{B}^{\infty }\frac{d\lambda _{1}}{\lambda _{1}}\int_{-\infty }^{\infty }%
\frac{d\lambda _{2}}{\sqrt{\lambda _{2}^{2}+1}\sqrt{(\lambda
_{1}+\lambda_{2})^{2}+1}}
\\
&\quad\quad\quad\quad\quad\quad\quad\quad\quad
=2\int_{B}^{\infty }\frac{d\lambda _{1}}{\lambda _{1}}%
\int_{0}^{\infty }\frac{d\lambda _{2}}{\sqrt{(\lambda _{2}-\frac{\lambda _{1}%
}{2})^{2}+1}\sqrt{(\lambda _{2}+\frac{\lambda _{1}}{2})^{2}+1}}
\\
&\leq 2\biggl (\int_{B}^{\infty }\frac{d\lambda _{1}}{\lambda _{1}}%
\int_{0}^{\infty }\frac{d\lambda _{2}}{((\lambda _{2}-\frac{\lambda _{1}}{2}%
)^{2}+1)((\lambda _{2}+\frac{\lambda _{1}}{2})^{2}+1)^{\frac{1}{4}}}\biggr )%
^{\frac{1}{2}}
\\
&\quad\quad\quad\quad\quad\quad\quad\quad\quad
\times\biggl (\int_{B}^{\infty }\frac{d\lambda _{1}}{\lambda _{1}}%
\int_{0}^{\infty }\frac{d\lambda _{2}}{((\lambda _{2}+\frac{\lambda _{1}}{2}%
)^{2}+1)^{\frac{3}{4}}}\biggr )^{\frac{1}{2}}<\infty.
\end{align*}%
This allows us to pass to limit in integral (\ref{Iv3}) by using
(\ref{faf0}) and ( \ref{[f]}):
\begin{align*}
&\mathrm{vH-}\lim I_{3}^{v}(E,t)
\\
&=\frac{\rho _{-\alpha
,-\alpha }(0)}{(2\pi)^2} \int_{-\infty}^\infty d\lambda _{2}\int_{-\infty}^\infty%
d\lambda _{1}e^{-i\tau \lambda _{1}}\frac{2\pi i\nu^\prime _{0}(E-2\alpha
s)(\lambda _{1}+2\pi i\nu^\prime _{0}(E-2\alpha s))}{\lambda
_{1}(\lambda _{1}+2\pi i(\nu^\prime _{0}(E)+\nu^\prime _{0}(E-2\alpha s))}
\\
&\quad\quad\quad\quad\quad\quad\quad\quad \times\frac{1}{\lambda
_{2}+f_{0}(E-2\alpha s-i0)}\cdot\frac{1}{\lambda _{1}+\lambda
_{2}+f_{0}(E-2\alpha s+i0)}.
\end{align*}%
Here integration path in $\lambda _{1}$ encircles zero from
above. Computing last integrals by residues we have
\begin{equation}
\mathrm{vH-}\lim I_{3}^{v}(E,t)=2\pi \rho _{-\alpha
,-\alpha }(0)\biggl[\frac{\nu^\prime _{0}(E-2\alpha s)}{\Gamma _{-\alpha }(E)}%
-e^{-\tau \Gamma _{-\alpha }(E)}\frac{\nu^\prime _{0}(E-2\alpha
s)}{\Gamma _{-\alpha }(E)}\biggr].  \label{I3}
\end{equation}%
Treating similarly the term $I_{2}^{v}$ in the r.h.s. of
(\ref{I123}) we obtain
\begin{align}
&\mathrm{vH-}\lim I_{2}^{v}(E,t)\label{I2}
\\
&=2\pi \rho _{\alpha ,\alpha }(0)\biggl[\frac{\nu^\prime _{0}(E)}{\Gamma
_{\alpha }(E)}+e^{-\tau \Gamma
_{\alpha }(E)}\frac{\nu^\prime _{0}(E+2\alpha s)}{\Gamma _{\alpha }(E)}\biggr]%
-\rho _{\alpha ,\alpha }(0)e^{-2\pi \tau \nu^\prime _{0}(E+2\alpha s)}.
\notag
\end{align}%
Now the assertion (\ref{rvHf}) of theorem follows from the (\ref{I123}), (\ref{I1}), (%
\ref{I3}) and (\ref{I2}).

\bigskip


\noindent Consider now the off-diagonal entry of (\ref{rrr}):
\begin{align}
&\rho _{\alpha ,-\alpha }(E,t) =\frac{\rho _{\alpha ,-\alpha}(0)}{(2\pi )^{2}}\biggl
[\int_{L_{2}^{v}}dz_{2}\int_{L_{1}^{v}}dz_{1}e^{-it\delta
z}f_{\alpha }(E,z_{1})f_{-\alpha }(E,z_{2})  \label{offd} \\
& +\int_{L_{2}^{v}}dz_{2}\int_{L_{1}^{v}}dz_{1}e^{-it\delta z}\frac{%
v^{4}f_{\alpha ,-\alpha }(z_{1},z_{2})f_{-\alpha ,\alpha
}(z_{1},z_{2})f_{\alpha }(E,z_{1})f_{-\alpha
}(E,z_{2})}{1-v^{4}f_{\alpha
,-\alpha }(z_{1},z_{2})f_{-\alpha ,\alpha }(z_{1},z_{2})}\biggr ]  \notag \\
& +\frac{\overline{\rho _{\alpha ,-\alpha}(0)}}{(2\pi )^{2}}
\int_{L_{2}^{v}}dz_{2}\int_{L_{1}^{v}}dz_{1}e^{-it\delta z}\frac{%
v^{2}f_{\alpha ,-\alpha }(z_{1},z_{2})f_{-\alpha
}(E,z_{1})f_{\alpha }(E,z_{2})}{1-v^{4}f_{\alpha ,-\alpha
}(z_{1},z_{2})f_{-\alpha ,\alpha
}(z_{1},z_{2})}  \notag \\
& =\frac{\rho_{\alpha ,-\alpha }(0)}{(2\pi )^{2}}%
\bigl [I_{1}^{v}(E,t)+I_{2}^{v}(E,t)\bigr ]+ \frac{\overline{\rho_{\alpha ,-\alpha }(0)}
}{(2\pi )^{2}}I_{3}^{v}(E,t).  \notag
\end{align}%
To find the limit of $I_{1}^{v}(E,t)$ of (\ref{offd}) we change
variables to $\zeta _{1}=v^{-2}(z_{1}-E_{\alpha })$, $\zeta
_{2}=v^{-2}(z_{2}-E_{-\alpha })$. This yields
\begin{equation}
I_{1}^{v}(E,t)=\exp (-2\alpha sit)\;J_{1}^{v},  \label{Iof1}
\end{equation}%
where
\begin{equation*}
J_{1}^{v}=\int_{L_{2}}d\zeta _{2}\frac{%
e^{i\tau \zeta _{2}}}{\zeta _{2}+f_{\alpha }(E_{-\alpha
}+v^{2}\zeta _{2})}\int_{L_{1}}d\zeta _{1}\frac{e^{-i\tau \zeta
_{1}}}{\zeta _{1}+f_{-\alpha }(E_{\alpha }+v^{2}\zeta _{1})},
\end{equation*}%
and by (\ref{faf0})
\begin{align}
\mathrm{vH-}\lim J_{1}^{v}(E,t)&= \int_{L_{2}}\frac{ e^{i\tau
\zeta _{2}}d\zeta _{2}}{\zeta _{2}+f_0(E-2\alpha
s-i0)}\int_{L_{1}}\frac{e^{-i\tau \zeta _{1}}d\zeta _{1}}{\zeta
_{1}+f_0(E+2\alpha s+i0)}\label{Jof1}
\\
&= (2\pi )^{2}\exp\{i\tau ((f_0(E+2\alpha s+i0)-f_0(E-2\alpha
s-i0))\}. \notag
\end{align}
We have similarly:
\begin{equation}
I_{2}^{v}=\exp (-2\alpha sit)\;J_{2}^{v},  \label{Iof2}
\end{equation}%
\begin{align*}
J_{2}^{v}=\int_{L_{1}}d\zeta _{1}\int_{L_{2}}d\zeta _{2}e^{-i\tau
\delta \zeta }&\frac{v^{4}f_{\alpha ,-\alpha }f_{-\alpha ,\alpha
}}{1-v^{4}f_{\alpha ,-\alpha }f_{-\alpha ,\alpha }}
\\
&\times\frac{1}{\zeta _{1}+f_{-\alpha }(E_{\alpha }+v^{2}\zeta
_{1})}\cdot \frac{1}{\zeta _{2}+f_{\alpha }(E_{-\alpha
}+v^{2}\zeta _{2})}.
\end{align*}%
Here we denote $f_{\beta,\gamma}=f_{\beta,\gamma}(E_{\alpha
}+v^{2}\zeta _{1},E_{-\alpha }+v^{2}\zeta _{2})$ (see
(\ref{f2})). By using the relations:
\begin{align*}
f_{\alpha ,-\alpha }=\frac{\delta f_{\alpha ,-\alpha
}}{v^{2}(\delta \zeta +\delta f_{-\alpha ,\alpha })},\quad
f_{-\alpha ,\alpha }=\frac{\delta f_{-\alpha ,\alpha }}{-4\alpha
s+v^{2}(\delta \zeta +\delta f_{\alpha ,-\alpha })},
\end{align*}%
where $\delta f_{\beta,\gamma}=f_{\beta}(E_{\alpha }+v^{2}\zeta _{1})
-f_{\gamma}(E_{-\alpha }+v^{2}\zeta _{2})$, we  obtain
\begin{align}
J_{2}^{v}=\int_{L_{2}}d\zeta _{2}\int_{L_{1}}d\zeta _{1}e^{-i\tau
\delta
\zeta }& \frac{v^{2}f_{-\alpha ,\alpha }\delta f_{\alpha ,-\alpha }}{%
1-v^{4}f_{\alpha ,-\alpha }f_{-\alpha ,\alpha
}}\cdot\frac{1}{\zeta _{1}-\zeta _{2}+\delta f_{-\alpha ,\alpha
}} \label{J2v}
\\
& \times\frac{1}{\zeta _{1}+f_{-\alpha }(E_{\alpha }+v^{2}\zeta _{1})}\cdot \frac{1}{%
\zeta _{2}+f_{\alpha }(E_{-\alpha }+v^{2}\zeta _{2})}.\notag
\end{align}%
Notice that
\begin{equation}
|f_{\alpha ,-\alpha }|\leq \frac{2c_{0}}{v^{2}(\eta _{1}+\eta
_{2})},\quad \alpha =\pm .  \notag
\end{equation}%
Hence
\begin{equation}
|1-v^{4}f_{\alpha ,-\alpha }f_{\alpha ,-\alpha }|\geq 1-\biggl(\frac{2c_{0}}{%
(\eta _{1}+\eta _{2})}\biggr)^{2}>\frac{1}{2},\;\text{if}\;\eta
_{j}\geq 2c_{0}, \notag
\end{equation}%
and integrand of (\ref{J2v}) is uniformly bounded from above by
integrable function $C(|\zeta _{1}-\zeta _{2}||\zeta _{1}||\zeta
_{2}|)^{-1}$. This allows us to pass to the limit in the integral
in (\ref{J2v}) and obtain that
\begin{align}
\mathrm{vH-}\lim J_{2}^{v}(E,t)=0.  \label{Jof2}
\end{align}
Treating similarly the term $I_{3}^{v}(E,t)$ of (\ref{offd}) we obtain
\begin{align}
\mathrm{vH-}\lim I_{3}^{v}(E,t)=0.  \label{Iof3}
\end{align}
Now assertion  (\ref{off}) of the theorem following from (\ref{offd})-(\ref{Iof3}).
\end{proof}

\medskip

According to (\ref{off}) the off-diagonal entry of the
reduced density matrix in the van Hove limit does not vanish but
just oscillates as const$\cdot e^{-2i\alpha st}$. The exponential
that determines these fast oscillations (recall that $t\rightarrow
\infty$) is the same as in the zero coupling
($\mathcal{S}_2$-isolated) limit of our model (\ref{Ham}), where
the reduced density matrix is
\begin{equation*}
\rho_{\alpha\beta}(E_k,t)\mid_{v^2=0}=e^{-its(\alpha-\beta)}\rho_{\alpha\beta}(0),
\end{equation*}
hence is again const$\cdot e^{-2i\alpha st}$ if $\alpha\neq\beta$,
($\alpha=-\beta$).

In the case where the two-level system models a continuous
quantum mechanical degree of freedom associated with a potential
with two wells (see e.g. \cite{Li-Co}   for examples and
discussion), the above oscillation reflects the phase coherence
between the quantum mechanical amplitudes for being in the left
and right wells, a pure quantum mechanical effect. In this case
our result means that an environment, modeled by a random matrix,
does not destroy the quantum mechanical coherence, at least in
the weak coupling regime corresponding to the van Hove limit.

However, from the statistical mechanics point of view the absence
of decay, moreover, fast oscillations, of the off-diagonal
entries of the reduced density matrix seems not too natural. In
this connection it worth noting that the fast ("microscopical")
oscillating behaviour of $\rho_{\alpha\beta}$, $\alpha\neq\beta$
can be converted into a decaying behaviour by several modification
of our initial setting.

One of them is to assume that the spacing $2s$ of our two-level
system is random and continuously distributed, although
concentrated around a certain $2s_0$. In other words, it is
necessary to assume that the two-level system is the subject of a
certain (even small) noise.

Another modification is to replace the van Hove limit
\begin{equation}
\lim_{t\rightarrow\infty}\rho(E,t)\mid_{v^2=\tau/t} \label{vHls}
\end{equation}
by
\begin{equation}
\lim_{t\rightarrow\infty,\triangle t\rightarrow\infty}(2\triangle
t)^{-1}\int_{t-\triangle t}^{t+\triangle
t}\rho(E,t\prime)\mid_{v^2=\tau/t\prime}dt\prime.\label{vHlc}
\end{equation}
If $\triangle t=t$, we just replace the limit
$t\rightarrow\infty$ by the Cesaro limit (time average limit), a
rather often used procedure in statistical mechanics. However the
off-diagonal entry vanishes even for $t\rightarrow\infty$, bat
$\triangle t/t\rightarrow 0$, although with a smaller rate of
decay. One can view this as an assumption on a sufficiently large
(macroscopic) measurent time: $s^{-1}<<\triangle t<<t$.

\bibliographystyle{amsalpha}

\end{document}